\documentclass{ws-m3as}
\let\trimmarks\relax
\usepackage[T1]{fontenc}
\usepackage{booktabs,tabularx,microtype}
\usepackage{algorithm,algpseudocode,enumitem,placeins,needspace}
\usepackage[super]{cite}
\usepackage{xcolor}
\usepackage[hypertexnames=false]{hyperref}
\hypersetup{colorlinks=true,citecolor=blue,linkcolor=blue,urlcolor=blue,
 pdftitle={A posteriori error bounds for cumulative energy fields in irreversible gradient damage},
 pdfauthor={Yubo Huang, Xiongyu Hu, Xin Lai, Cheng Lyu}}
\makeatletter
\let\c@lemma\c@theorem
\let\c@proposition\c@theorem
\let\c@corollary\c@theorem
\let\c@definition\c@theorem
\let\c@remark\c@theorem
\makeatother

\newcommand{\dd}{\,\mathrm{d}}
\newcommand{\raw}{\mathrm{raw}}
\newcommand{\refe}{\mathrm{ref}}
\newcommand{\TV}{\operatorname{TV}}

\newcommand{\norm}[1]{\left\|#1\right\|}

\newcommand{\LL}{\mathcal L}
\newcommand{\EE}{\mathcal E}

\setlist[itemize]{leftmargin=*,itemsep=2pt,topsep=4pt}
\setlist[enumerate]{leftmargin=*,itemsep=2pt,topsep=4pt}
\allowdisplaybreaks[2]
\begin{document}
\markboth{Y. Huang, X. Hu, X. Lai and C. Lyu}{Error bounds for cumulative energy fields}
\title{A posteriori error bounds for cumulative energy fields\\
in irreversible gradient damage}
\author{Yubo Huang, Xiongyu Hu\footnote{Corresponding author.}, and Xin Lai}
\address{Key Laboratory of Transportation Tunnel Engineering, Ministry of Education\\
Southwest Jiaotong University, Chengdu, China\\
\href{mailto:ybforever@my.swjtu.edu.cn}{ybforever@my.swjtu.edu.cn}\\
\href{mailto:huxiongyu@swjtu.edu.cn}{huxiongyu@swjtu.edu.cn}\\
\href{mailto:xinlai@my.swjtu.edu.cn}{xinlai@my.swjtu.edu.cn}}
\author{Cheng Lyu}
\address{Faculty of Geosciences and Engineering\\
Southwest Jiaotong University, Chengdu, China\\
\href{mailto:lvchengaust@163.com}{lvchengaust@163.com}}
\maketitle
\begin{abstract}
Cumulative energy fields in irreversible evolution depend on the loading path, even when terminal states and global work agree. We derive a posteriori bounds for these fields in scalar gradient damage with a prescribed driving density. A perturbation estimate for the incremental obstacle map yields a continuous loading history, which is enclosed by nondecreasing residual supersolutions. A primitive identity transfers this enclosure to the cumulative field. An exact endpoint-rule defect and outward arithmetic certify the output delivered by the numerical method. The analysis separates spatial, loading, and algebraic errors and gives sufficient conditions for certificate decay. Computable realizations use splines and verified supersolutions in one dimension and on two-dimensional rectangles. The production bounds cover the independently enclosed errors of all 417 saved histories from 28 adaptive calculations and all 25 fixed-parameter controls. In a narrow-localization case, the verified relative error decreases from approximately 158.78\% to below 2.125\% at a 5\% certified target. Source-oriented marking reduces accumulated matrix--vector work by 35.9\% and 51.9\% in two comparisons. A mechanically generated two-dimensional benchmark shows how auxiliary supersolutions sharpen constant barriers and enable acceptance of both loading orders.
\end{abstract}
\keywords{A posteriori error bounds; variational inequalities; irreversible evolution;
cumulative energy fields; admissible reconstruction; verified computation.}
\ccode{AMS Subject Classification: 65N15, 65N30, 49J40, 65G50}
\section{Introduction}\label{snum:introduction}
Gradient damage models use an evolving internal field and a regularized
spatial interaction to describe material degradation
\cite{Peerlings1996}. Variational fracture \cite{Francfort1998} and
phase-field regularizations \cite{Zhuang2022} provide a related energetic
description of crack evolution, while rate-independent
evolution supplies a mathematical setting for irreversibility
\cite{Mielke2015}. Along a damage history, the local work of the degradation
force defines a cumulative energy field. Spatial coupling makes its
distribution sensitive to the relative progress of neighboring regions.
Computing this field accurately therefore requires control of the
irreversible path and of the numerical accumulation performed along it.

Residual a posteriori estimates relate approximation errors to computable
defects of the governing equations \cite{Chamoin2023}.
For elliptic obstacle problems, maximum-norm estimates and
contact-sensitive barriers provide pointwise control
\cite{Nochetto2003,Nochetto2005}. Convex duality and explicit flux
reconstruction yield error representations for nonsmooth convex
variational problems \cite{BartelsKaltenbach2024}, including error
identities for scalar Signorini contact \cite{Bartels2025}.
Calibration-based estimates address nonconvex variational problems
\cite{Berkels2026}. For evolution problems, elliptic reconstruction
connects continuous stability with discretization defects
\cite{Makridakis2003}. This approach has recently been combined with
conditional stability estimates for the Allen--Cahn equation with
variable mobility \cite{Brunk2025}.

Goal-oriented analysis weights residuals by dual solutions to control
selected outputs \cite{Becker2001}. Recent dual-weighted residual (DWR)
developments treat
nonlinear and multiple outputs and balance discretization with iteration
error under saturation hypotheses \cite{Endtmayer2020,Endtmayer2024}.
Discrete adjoint estimates can retain nonlinear remainder terms
\cite{Granzow2023}, while space--time localization provides separate
temporal and spatial indicators for parabolic problems \cite{Thiele2024}.
Strict bounds also cover history-dependent local quantities in
viscoelasticity \cite{Chamoin2007} and combine spatial, temporal, and
iterative errors in quasi-static plasticity \cite{Ladeveze2012}.

Adaptive methods for damage and fracture address localization and the
cost of incremental solution. Goal-oriented error assessment has been
developed for gradient-enhanced damage \cite{Pannachet2010} and
phase-field fracture \cite{Wick2016}. Residual estimates account for
irreversibility constraints \cite{Walloth2022}, and adaptive quasi-Newton
stepping controls incremental solution \cite{Kristensen2020}.
Spatial and temporal refinement are combined in dynamic fracture
\cite{Labanda2022} and in energy-based adaptive calculations for
heterogeneous materials \cite{Rohracker2026}. Residual minimization in a
discontinuous-Galerkin dual norm also supplies mesh indicators
\cite{Liu2026}. For strongly monotone nonlinear elliptic problems, coupled control of
discretization, linearization, and algebraic errors yields adaptive
algorithms with quasi-optimal cost \cite{Haberl2021}. Equilibrated fluxes
give guaranteed upper bounds and degree-robust local efficiency for
isogeometric Poisson approximations \cite{Gantner2024}.

The problem studied here couples three sources of error through an
evolving obstacle. An approximate state determines the admissible set of
the next increment, loading discretization changes the path through these
sets, and the endpoint accumulation introduces an additional output
defect. We treat their combined effect on a spatially distributed
cumulative field. The central construction is a nondecreasing
supersolution that propagates pointwise uncertainty through successive
obstacles. A primitive identity then transfers the resulting history
enclosure to the source field.

We establish this construction for scalar irreversible gradient damage
with a prescribed driving density. The analysis gives a continuous
loading-history limit, a causal comparison estimate, and a relative
$L^2$ bound for the delivered cumulative field. The bound includes the
actual endpoint-rule defect and outward arithmetic, with spatial,
loading, and algebraic contributions available for accuracy allocation.
The comparison and source-transfer results are formulated on general
bounded Lipschitz domains. Computable realizations use admissible spline
reconstruction and verified supersolutions in one dimension and on
two-dimensional rectangles. Sufficient conditions for certificate decay complement these bounds.
Numerical assessment covers 417 saved histories from 28 adaptive calculations,
25 fixed-parameter controls, and a mechanically generated two-dimensional
damage history. It quantifies reliability and the benefits of causal
bounds and source-oriented marking.

Section~\ref{sec:model} defines the irreversible problem and cumulative
field. Section~\ref{sec:certificate} establishes the error bounds, and
Section~\ref{sec:realizations} gives their computable realizations.
Sections~\ref{sec:algorithm} and \ref{sec:numerical} present adaptive
allocation and numerical assessment. Section~\ref{snum:conclusions}
concludes the paper.

\section{The Irreversible Problem and Its Cumulative Energy Field}\label{sec:model}
Let $\Omega\subset\mathbb R^d$ be a bounded Lipschitz domain, $a,b>0$,
and $0\leq f\in W^{1,\infty}(0,T;L^\infty(\Omega))$. Set
$L=\|f_t\|_{L^\infty(\Omega\times(0,T))}$; the continuous-in-loading
representative of $f$ is used throughout. Define
\begin{equation}\label{eq:energy}
 \EE_t(v)=\frac12\int_\Omega
 \big[f(x,t)(1-v)^2+a v^2+b|\nabla v|^2\big]\dd x,
 \qquad v\in H^1(\Omega).
\end{equation}
The natural boundary condition is homogeneous Neumann. The interaction
length is $\sqrt{b/a}$. The initial field $d_0\in H^1(\Omega)$ lies
between zero and one and is stable at $f(0)$, in the sense of the increment
below with obstacle $d_0$.

For a partition $\pi=\{0=t_0<\cdots<t_N=T\}$, define
\begin{equation}\label{eq:increment}
 d_n^\pi=\operatorname*{arg\,min}_{v\in K(d_{n-1}^\pi)}\EE_{t_n}(v),
 \qquad K(w)=\{v\in H^1(\Omega):w\leq v\leq1\}.
\end{equation}
All inequalities between fields are almost-everywhere inequalities.
Writing
\begin{equation}\label{eq:operator}
 \LL_t=-b\Delta+a+f(\cdot,t),
\end{equation}
the weak optimality condition is
\begin{equation}\label{eq:vi}
 \int_\Omega\!\big[b\nabla d_n^\pi\cdot\nabla(v-d_n^\pi)
       +\big((a+f_n)d_n^\pi-f_n\big)(v-d_n^\pi)\big]\dd x\geq0,
 \quad v\in K(d_{n-1}^\pi).
\end{equation}
The continuous target $d$ is the limit of these histories as $|\pi|\to0$.
Proposition~\ref{prop:history-limit} proves that this limit exists,
is partition-independent, and has pointwise nondecreasing Lipschitz loading
paths. The cumulative energy field is
\begin{equation}\label{eq:source}
 Q[d](x,T)=\int_0^T f(x,t)(1-d(x,t))d_t(x,t)\dd t.
\end{equation}
The force $f(1-d)$ is conjugate to damage in \eqref{eq:energy}, so
$Q[d]$ is its accumulated local work. We call this output the cumulative
source.

The numerical output is specified separately. Let $0\leq u_{n-1}\leq u_n
\leq1$ be the stored spatial fields, with $u_0=d_0$, and let $u$ interpolate
them affinely in loading. The exact right-endpoint polynomial formed from
these stored values is
\begin{equation}\label{eq:raw}
 \widetilde Q_{\raw}(x,T)
 =\sum_{n=1}^N f_n(x)(1-u_n(x))(u_n(x)-u_{n-1}(x)).
\end{equation}
Its stored floating-point coefficient field is $Q_{\raw}$, with arithmetic
error $Q_{\raw}-\widetilde Q_{\raw}$. Given $\varepsilon>0$, we seek
\begin{equation}\label{eq:goal}
 \frac{\|Q_{\raw}-Q[d]\|_{L^2(\Omega)}}{\|Q[d]\|_{L^2(\Omega)}}
 \leq\varepsilon,
\end{equation}
with a computed positive lower bound for the exact source norm.

A local work identity makes the loading-path contribution explicit. For a
smooth stationary monotone history satisfying $\LL_t d=f$, testing with
$d_t$ gives
\begin{equation}\label{snum:global-work}
 \int_\Omega Q[d]\dd x
 =\frac12\left[\int_\Omega(a d^2+b|\nabla d|^2)\dd x\right]_0^T.
\end{equation}
For such histories, the corresponding distributional local identity is
\begin{equation}\label{snum:local-work}
 Q[d]=\frac12[a d^2+b|\nabla d|^2]_0^T
       -b\nabla\cdot\int_0^T d_t\nabla d\dd t.
\end{equation}
Section~\ref{sec:2d-mechanism} proves these identities with the required
$H^1$ loading regularity. Common endpoints fix the global work, while
the accumulated flux retains local loading-order effects. The error
analysis also allows nonmonotone drives and active obstacles.
\subsection{Loading order and local accumulated work}\label{supp:motivation}
Consider a finite-dimensional counterpart with symmetric coupling matrix
$K$. During growth of component
$i$, equilibrium gives $Y_i=(Kd)_i$; during inactivity, $\dot d_i=0$.
Consequently $q_i=\int_0^T(Kd)_i\dot d_i\dd t$. The product rule gives
\[
 2(Kd)_i\dot d_i
 =\frac{\dd}{\dd t}\big[d_i(Kd)_i\big]
   +\sum_jK_{ij}(d_j\dot d_i-d_i\dot d_j).
\]
Integration along the loading path gives
\begin{equation}\label{eq:path-area}
 q_i=\frac12\left[d_i(Kd)_i\right]_0^T+
 \frac12\sum_j K_{ij}\int_0^T
 \left(d_j\,\dot d_i-d_i\,\dot d_j\right)\dd t .
\end{equation}
The second term contains oriented areas in damage space. Symmetry gives
$\sum_iq_i=\tfrac12[d^\top Kd]_0^T$; diagonal $K$ removes the area
terms altogether. Off-diagonal coupling permits different local sources
at identical endpoints and global work. Figure~\ref{fig:history} shows
the two-component example of Eq.~\eqref{eq:counterexample} and its
spatially distributed counterpart.

\begin{figure}
\centering\includegraphics[width=\linewidth]{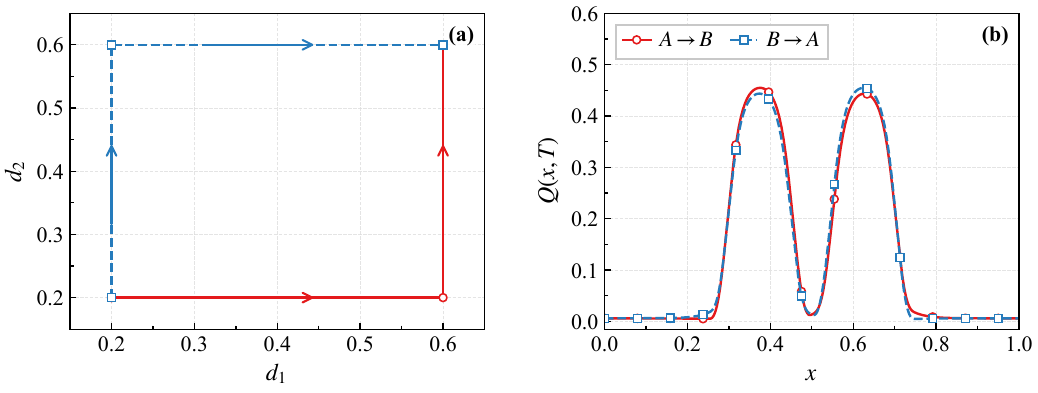}
\caption{Loading-order dependence: (a) irreversible paths for the two-component example
in Eq.~\eqref{eq:counterexample}; (b) independently enclosed one-dimensional
source profiles. Both comparisons have identical terminal damage and different local sources.}\label{fig:history}
\end{figure}

\Needspace{10\baselineskip}

\subsection{Positioning with respect to a posteriori output control}\label{sec:positioning}

Table~\ref{tab:positioning} summarizes the relation to existing error
control. The evolving obstacle requires a pointwise history enclosure,
provided by Theorem~\ref{thm:history}. Theorems~\ref{thm:transfer} and
\ref{thm:source} then transfer it to the cumulative field and include
endpoint accumulation and arithmetic.

\begin{table}
\caption{Representative error-control results and their relation to the present analysis. Guarantees refer to the assumptions of the cited formulations.}\label{tab:positioning}
\centering\footnotesize
\setlength{\tabcolsep}{4pt}
\begin{tabularx}{\linewidth}{@{}>{\raggedright\arraybackslash}p{.21\linewidth}>{\raggedright\arraybackslash}X>{\raggedright\arraybackslash}X@{}}
\toprule
Formulation & Established error-control object & Relation to the present problem\\
\midrule
Obstacle residuals and barriers \cite{Nochetto2003,Nochetto2005}
 & Pointwise control and localized estimates for elliptic contact problems
 & Supply the comparison background; the obstacle here also inherits loading-history uncertainty.\\
Reconstruction-based evolution estimates \cite{Makridakis2003,Brunk2025}
 & Continuous stability combined with computable defects of an approximation
 & Motivate separating the history stability theorem from its numerical realization.\\
DWR \cite{Endtmayer2020,Endtmayer2024,Thiele2024}
 & Single and multiple output errors; spatial, temporal, and iteration control, with guarantees under the respective hypotheses
 & Provides the output-oriented perspective; the present pointwise history enclosure is transferred to the full cumulative field.\\
Strict output bounds \cite{Chamoin2007,Ladeveze2012}
 & Guaranteed bounds for selected local outputs, including history-dependent and inelastic settings
 & Connects history-dependent error control to local outputs; the present bound combines causal obstacle inheritance with the endpoint rule.\\
Damage and fracture adaptivity \cite{Walloth2022,Labanda2022,Rohracker2026,Liu2026}
 & Irreversibility-constrained residuals, spatial and temporal indicators, or residual-minimization indicators
 & Motivates the separate spatial, loading, and solver channels; the present target is a prescribed-drive cumulative field.\\
Present framework
 & Pointwise source bounds and a relative $L^2$ certificate for the delivered field
 & Includes inherited history uncertainty, reconstruction, endpoint accumulation, and arithmetic. Full adaptive allocation is assessed in one dimension; two dimensions use nested grids.\\
\bottomrule
\end{tabularx}
\end{table}

\section{General Causal History-to-Output Certification}\label{sec:comparison}\label{sec:certificate}

We first bound the irreversible history by scalar comparison and then
transfer this bound to the cumulative field. The endpoint-rule defect
completes the estimate for the delivered output.

\subsection{Scalar obstacle comparison and the continuous history}
Throughout this section, $\Omega\subset\mathbb R^d$ is a bounded Lipschitz
domain, $a,b>0$, and $f\in W^{1,\infty}(0,T;L^\infty(\Omega))$ is
nonnegative. Order relations mean almost-everywhere order. The initial
state belongs to $H^1(\Omega)\cap L^\infty(\Omega)$, lies in $[0,1]$,
and is stable at the initial drive. We use the weak Neumann realization
of $\LL_f=-b\Delta+a+f$, interpreting the boundary condition through
$H^1$ test functions. To make each comparison explicit, write
\begin{equation}\label{eq:weak-form-definition}
 A_f(v,\phi)=\int_\Omega\{b\nabla v\cdot\nabla\phi+(a+f)v\phi\}\dd x,
 \qquad \|v\|_0^2=b\|\nabla v\|_2^2+a\|v\|_2^2.
\end{equation}
Thus $\LL_fv\geq h$ means $A_f(v,\phi)\geq(h,\phi)$ for every
nonnegative $\phi\in H^1(\Omega)$. The norm $\|\cdot\|_0$ is equivalent
to the $H^1$ norm because both $a$ and $b$ are positive. The identity
$\nabla(v_+)=\mathbf1_{\{v>0\}}\nabla v$ is understood in the Sobolev
sense. It justifies the positive-part comparisons below for Sobolev functions
and their measurable contact sets.

\begin{lemma}[Incremental comparison and uniform perturbations]\label{lem:stability}
For $0\leq w\leq1$ in $H^1(\Omega)$, let $S(f,w)$ denote the unique
minimizer of Eq.~\eqref{eq:energy} over $v\geq w$. This minimizer lies
in $[0,1]$ and hence also minimizes over $K(w)$. The map is increasing
in both nonnegative $f$ and $w$, and $S(f,S(f,w))=S(f,w)$. Moreover,
\begin{equation}\label{eq:increment-stability}
 \|S(f,w)-S(g,z)\|_{L^\infty}
 \leq \max\{\|w-z\|_{L^\infty},a^{-1}\|f-g\|_{L^\infty}\}.
\end{equation}
These statements hold in every spatial dimension under the stated assumptions.
\end{lemma}
\begin{proof}
The set $K_+(w)=\{v\in H^1:v\geq w\}$ is nonempty, convex, and
weakly closed. The energy is weakly lower semicontinuous and satisfies
\begin{equation}\label{eq:coercive-energy}
 \EE_f(v)=\tfrac12 A_f(v,v)-(f,v)+\tfrac12\int_\Omega f\dd x
 \geq\tfrac12\|v\|_0^2.
\end{equation}
The direct method and strict convexity give a unique minimizer $u$.
Its variational inequality is $A_f(u,v-u)\geq(f,v-u)$ for
$v\in K_+(w)$. With $\eta=(u-1)_+$, the admissible test
$v=u-\eta=\min(u,1)$ gives
\begin{align}
 0&\geq A_f(u,\eta)-(f,\eta)\notag\\
 &=b\|\nabla\eta\|_2^2+
   \int_{\{u>1\}}[(a+f)(1+\eta)-f]\eta\dd x\notag\\
 &=\|\eta\|_0^2+\int_\Omega f\eta^2\dd x+a\int_\Omega\eta\dd x
 \geq\|\eta\|_0^2.\label{eq:range-test}
\end{align}
Thus $w\leq u\leq1$, and $u$ also minimizes over $K(w)$.

Testing with $u+\phi$, $\phi\geq0$, shows that $u$ is a weak
supersolution. If $s\geq w$ is another supersolution, set
$\eta=(u-s)_+$ and use $u-\eta=\min(u,s)\in K_+(w)$. Then
\begin{align}
 0\leq\|\eta\|_0^2
 &\leq A_f(u-s,\eta)
 =A_f(u,\eta)-A_f(s,\eta)\notag\\
 &\leq(f,\eta)-(f,\eta)=0.\label{eq:least-super-test}
\end{align}
Hence $u\leq s$, which is the least-supersolution characterization.

For $f\leq g$ and $w\leq z$, put $v=S(g,z)$. Since $0\leq v\leq1$
and $v\geq z\geq w$, every nonnegative test satisfies
\begin{align}
 A_f(v,\phi)-(f,\phi)
 &=A_g(v,\phi)-(g,\phi)+((g-f)(1-v),\phi)\notag\\
 &\geq((g-f)(1-v),\phi)\geq0.
 \label{eq:drive-order-chain}
\end{align}
Comparison gives $S(f,w)\leq S(g,z)$. For $u=S(f,w)$, admissibility
and the same characterization give
\[
 u\leq S(f,u)\leq u,
 \qquad S(f,S(f,w))=S(f,w).
\]

Finally, let
$\delta=\max\{\|w-z\|_\infty,a^{-1}\|f-g\|_\infty\}$ and
$v=S(g,z)$. The shifted field satisfies $v+\delta\geq z+\delta\geq w$,
and
\begin{align}
 A_f(v+\delta,\phi)-(f,\phi)
 &=A_g(v,\phi)-(g,\phi)
   +((f-g)(v-1)+(a+f)\delta,\phi)\notag\\
 &\geq\int_\Omega[-\|f-g\|_\infty+a\delta]\phi\dd x
 \geq0.\label{eq:constant-shift-chain}
\end{align}
Applying comparison in both directions yields
\[
 S(f,w)\leq S(g,z)+\delta,
 \qquad S(g,z)\leq S(f,w)+\delta,
\]
which proves \eqref{eq:increment-stability}.
\end{proof}

\begin{proposition}[Continuous loading-history limit]\label{prop:history-limit}
The right-continuous incremental histories converge uniformly in
$L^\infty(\Omega)$ as the maximum loading step tends to zero, independently
of the partitions. Their limit $d$ is nondecreasing and satisfies
$0\leq d\leq1$. With $L=\|f_t\|_{L^\infty(\Omega\times(0,T))}$,
$\|d(t)-d(s)\|_\infty\leq L|t-s|/a$.
\end{proposition}
\begin{proof}
For $t\in[t_n,t_{n+1})$, set $d^\pi(t)=d_n^\pi$ and
$f^\pi(t)=f(t_n)$, using the terminal values at $T$. Then
$\|f^\pi-f\|_{L^\infty(\Omega\times[0,T])}\leq L|\pi|$.
On the ordered union $0=s_0<\cdots<s_J=T$ of two partitions,
idempotence permits insertion of updates with unchanged drive.
For $D_j=\|d^\pi(s_j)-d^{\pi'}(s_j)\|_\infty$,
Lemma~\ref{lem:stability} gives $D_0=0$ and
\begin{align}
 D_j
 &\leq\max\{D_{j-1},a^{-1}\|f^\pi(s_j)-f^{\pi'}(s_j)\|_\infty\}
 \notag\\
 &\leq a^{-1}\max_{0\leq i\leq j}
           \|f^\pi(s_i)-f^{\pi'}(s_i)\|_\infty
 \leq\frac{L}{a}(|\pi|+|\pi'|).
 \label{eq:common-partition-chain}
\end{align}
Both histories are constant between union points, so
\begin{equation}\label{eq:uniform-history-cauchy}
 \sup_{0\leq t\leq T}\|d^\pi(t)-d^{\pi'}(t)\|_\infty
 \leq\frac{L}{a}(|\pi|+|\pi'|).
\end{equation}
Completeness gives a partition-independent uniform limit $d$.

Using the preceding state as a competitor gives
\begin{align}
 \EE_{t_n}(d_n^\pi)
 &\leq\EE_{t_n}(d_{n-1}^\pi)\notag\\
 &=\EE_{t_{n-1}}(d_{n-1}^\pi)
   +\tfrac12\int_\Omega(f_n-f_{n-1})(1-d_{n-1}^\pi)^2\dd x\notag\\
 &\leq\EE_{t_{n-1}}(d_{n-1}^\pi)
   +\tfrac12|\Omega|L\Delta t_n.
 \label{eq:history-energy-step}
\end{align}
Summation and \eqref{eq:coercive-energy} imply
\[
 \tfrac12\|d^\pi(t)\|_0^2
 \leq\EE_{t_n}(d_n^\pi)
 \leq\EE_0(d_0)+\tfrac12|\Omega|LT,
 \qquad t\in[t_n,t_{n+1}).
\]
At each fixed $t$, a weakly convergent $H^1$ subsequence has limit $d(t)$,
since uniform $L^\infty$ convergence also gives strong $L^2$ convergence.
Thus $d(t)\in H^1$. Put $\epsilon_\pi(t)=\|d^\pi(t)-d(t)\|_\infty$.
Idempotence and Lemma~\ref{lem:stability} give
\begin{align}
 &\|S(f(t),d(t))-d(t)\|_\infty\notag\\
 &\quad\leq\|S(f(t),d(t))-S(f^\pi(t),d^\pi(t))\|_\infty
       +\epsilon_\pi(t)\notag\\
 &\quad\leq\max\{\epsilon_\pi(t),L|\pi|/a\}
       +\epsilon_\pi(t)\longrightarrow0.
 \label{eq:limit-stability-chain}
\end{align}
This proves $S(f(t),d(t))=d(t)$.

For $s\leq t$, choose $t_i\leq s<t_{i+1}$ and
$t_j\leq t<t_{j+1}$, with the usual terminal convention.
Compare updates $i+1,\ldots,j$ with the constant sequence $d_i^\pi$,
which is stable at $f(t_i)$. For $j>i$, the same recursion gives
\begin{align*}
 \|d_j^\pi-d_i^\pi\|_\infty
 &\leq a^{-1}\max_{i<k\leq j}\|f(t_k)-f(t_i)\|_\infty\\
 &\leq L(t_j-t_i)/a\leq L(t-s+|\pi|)/a,\\
 \|d(t)-d(s)\|_\infty
 &\leq2\|d-d^\pi\|_{L^\infty(\Omega\times[0,T])}
       +L(t-s+|\pi|)/a.
\end{align*}
The case $j=i$ has zero discrete increment. Letting $|\pi|\to0$
proves the Lipschitz estimate; order and range pass through uniform
convergence. Choose representatives at rational loading times and remove
the countable union of the null sets in these inequalities. At every
remaining spatial point they have a unique nondecreasing Lipschitz
extension, representing $d(t)$ for each $t$, with
\[
 d(x,\cdot)\in W^{1,\infty}(0,T),
 \qquad 0\leq d_t(x,t)\leq L/a\quad\text{for a.e. }t.
\]
\end{proof}

\subsection{A certified history-enclosure provider}\label{sec:barriers}
Let $p$ be continuous and affine in loading on each interval
$I_n=(t_{n-1},t_n)$, with values in $H^1(\Omega)$, zero weak boundary
flux, $0\leq p\leq1$, $p_t\geq0$, and $p(0)=d_0$.
For a strong-residual realization we additionally require
$\Delta p\in L^2(\Omega)$; a globally $C^1$, piecewise $H^2$
reconstruction is sufficient. Suppose $\lambda,r\in L^2(\Omega)$ satisfy
\begin{equation}\label{eq:complementary-residual}
 r=f+\lambda-\LL_t p,\qquad\lambda\geq0,\qquad\lambda p_t=0.
\end{equation}
The residual relation, complementarity and the inequalities below are
required for almost every loading time; spatial inequalities are weak
inequalities against nonnegative $H^1$ test functions. The proof uses
continuity of the reconstruction and driving density to extend the weak
inequalities to comparison times. The multiplier enters through its
almost-everywhere complementarity condition.

\begin{theorem}[Causal history enclosure]\label{thm:history}
Suppose $e_n\in H^1(\Omega)\cap L^\infty(\Omega)$ are nonnegative and
\begin{equation}\label{eq:barrier-condition}
 e_n\geq e_{n-1},\qquad \LL_t e_n\geq|r(t)|\quad\text{for a.e. }t\in I_n,
 \qquad e_0=0.
\end{equation}
Then the continuous history of Proposition~\ref{prop:history-limit} satisfies
\begin{equation}\label{eq:history-bound}
 |d(x,t)-p(x,t)|\leq e_n(x)\quad\text{a.e. in }\Omega\times I_n.
\end{equation}
Section~\ref{sec:realizations} constructs supersolutions satisfying these
hypotheses and encloses their values in verified arithmetic.
\end{theorem}
\begin{proof}
Fix $I_n$ and define $A_n=\{p_n-p_{n-1}>0\}$. For a nonnegative
$\phi\in H^1$, the residual and supersolution inequalities give
\begin{align}
 A_{f(t)}(p+e_n,\phi)-(f(t),\phi)
 &=(\lambda-r,\phi)+A_{f(t)}(e_n,\phi)\notag\\
 &\geq(\lambda+|r|-r,\phi)\geq0.
 \label{eq:upper-barrier-chain}
\end{align}
On $A_n$, affine interpolation gives $p_t=(p_n-p_{n-1})/\Delta t_n>0$,
so complementarity yields $\lambda=0$. For nonnegative $\phi\in H^1$
vanishing outside $A_n$,
\begin{align}
 A_{f(t)}(p-e_n,\phi)-(f(t),\phi)
 &=(\lambda-r,\phi)-A_{f(t)}(e_n,\phi)\notag\\
 &\leq(-r-|r|,\phi)\leq0.
 \label{eq:lower-barrier-chain}
\end{align}
For each fixed test the left sides are continuous on
$[t_{n-1},t_n]$, since $p$ is affine in $H^1$ and $f$ is continuous
in $L^\infty$. Thus both inequalities extend to the segment endpoints.

Choose a comparison partition containing all knots of $p$. Suppose that
successive times $s_{k-1}<s_k$ lie in this segment and that
\[
 p(s_{k-1})-e_n\leq d_{k-1}^\pi\leq p(s_{k-1})+e_n.
\]
Write $p_k=p(s_k)$ and $f_k=f(s_k)$. Monotonicity of $p$ gives
\begin{equation}\label{eq:upper-obstacle-inheritance}
 p_k+e_n\geq p(s_{k-1})+e_n\geq d_{k-1}^\pi,
 \qquad \LL_{f_k}(p_k+e_n)\geq f_k.
\end{equation}
Set $\eta=(d_k^\pi-p_k-e_n)_+$. The test
$d_k^\pi-\eta=\min(d_k^\pi,p_k+e_n)$ lies above the inherited obstacle.
Using its variational inequality and \eqref{eq:upper-barrier-chain},
\begin{align*}
 0\leq\|\eta\|_0^2
 &\leq A_{f_k}(d_k^\pi-p_k-e_n,\eta)\\
 &\leq(f_k,\eta)-A_{f_k}(p_k+e_n,\eta)\leq0.
\end{align*}
Hence $d_k^\pi\leq p_k+e_n$.

For the lower bound, let $v=p_k-e_n$ and $w=(v-d_k^\pi)_+$.
On $\{w>0\}$ the inherited enclosure implies
\begin{align}
 p_k-e_n&>d_k^\pi\geq d_{k-1}^\pi
                    \geq p(s_{k-1})-e_n,\notag\\
 0&<p_k-p(s_{k-1})
   =\frac{s_k-s_{k-1}}{\Delta t_n}(p_n-p_{n-1}).
 \label{eq:lower-growth-support}
\end{align}
Thus $w$ vanishes outside $A_n$ and is admissible in
\eqref{eq:lower-barrier-chain}. The exact obstacle solution is a
supersolution by Lemma~\ref{lem:stability}, and therefore
\begin{align}
 0\leq\|w\|_0^2
 &\leq A_{f_k}(v-d_k^\pi,w)
 =A_{f_k}(v,w)-A_{f_k}(d_k^\pi,w)\notag\\
 &\leq A_{f_k}(v,w)-(f_k,w)\leq0.
 \label{eq:lower-coercivity-chain}
\end{align}
The last step is \eqref{eq:lower-barrier-chain}, including its endpoint
extension. Consequently $w=0$ and $p_k-e_n\leq d_k^\pi$.

The induction starts from $p(0)=d_0$. At a reconstruction knot,
\[
 p_n-e_{n+1}\leq p_n-e_n\leq d^\pi(t_n)
 \leq p_n+e_n\leq p_n+e_{n+1},
\]
so $e_{n+1}\geq e_n$ initializes the next segment. For a fixed
$t\in[t_{n-1},t_n]$, take comparison partitions containing $t$ and the
reconstruction knots. Proposition~\ref{prop:history-limit} gives
\[
 \big\|\big(|d(t)-p(t)|-e_n\big)_+\big\|_\infty
 \leq\|d(t)-d^\pi(t)\|_\infty\longrightarrow0.
\]
The enclosure on a countable dense set of times extends by loading
continuity and holds almost everywhere on $\Omega\times I_n$.
\end{proof}

Computable cell bounds enclose an $H^1$ supersolution satisfying the
causal order. A constant-in-space construction follows immediately.

\begin{proposition}[Constant multidimensional realization]\label{prop:constant}
If $a+f\geq a_0>0$ and $|r|\leq\sum_j\rho_{j,n}$ on
$\Omega\times I_n$, where each $\rho_{j,n}$ is a certified nonnegative
number, set
\begin{equation}\label{eq:constant-barrier}
 e_{j,n}=a_0^{-1}\max_{m\leq n}\rho_{j,m},\qquad e_n=\sum_j e_{j,n}.
\end{equation}
These radii satisfy Theorem~\ref{thm:history} in every dimension.
\end{proposition}
\begin{proof}
With $R_{j,n}=\max_{m\leq n}\rho_{j,m}$,
\[
 R_{j,n}=\max\{R_{j,n-1},\rho_{j,n}\},
 \qquad e_{j,n}\geq e_{j,n-1}\geq0,
 \qquad a_0e_{j,n}=R_{j,n}\geq\rho_{j,n}.
\]
These spatial constants belong to $H^1\cap L^\infty$ and have zero
gradient. For every nonnegative $\phi\in H^1$ and a.e. $t\in I_n$,
\begin{align}
 A_{f(t)}(e_n,\phi)
 &=\int_\Omega(a+f(t))\sum_j\frac{R_{j,n}}{a_0}\phi\dd x\notag\\
 &\geq\int_\Omega\sum_jR_{j,n}\phi\dd x
 \geq\int_\Omega\sum_j\rho_{j,n}\phi\dd x
 \geq\int_\Omega|r(t)|\phi\dd x.
 \label{eq:constant-provider-chain}
\end{align}
Theorem~\ref{thm:history} applies with $e_0=0$.
\end{proof}

\subsection{Field-valued source transfer and the delivered output}\label{sec:source-transfer}
The source-transfer argument uses a pointwise history bound and is
independent of its construction.
Let $d,u$ be measurable $[0,1]$-valued fields with absolutely continuous
loading paths almost everywhere, and let $f(x,\cdot)\in W^{1,1}(0,T)$.
Assume the products in the following formulas are integrable. These
requirements hold for the bounded data and histories considered here.
Set $F(z)=z-z^2/2$ and
\begin{equation}\label{eq:F-bound}
 H(u,E)=\min\{E,(1-u)E+E^2/2\}.
\end{equation}

\begin{theorem}[Dimension-independent history-to-source transfer]\label{thm:transfer}
Suppose $|d(x,t)-u(x,t)|\leq E(x,t)$ almost everywhere, with the same
bound at $0$ and $T$. Then
\begin{equation}\label{eq:primitive}
 Q[d]=f(T)F(d(T))-f(0)F(d(0))-\int_0^T f_tF(d)\dd t
\end{equation}
and $|Q[d]-Q[u]|\leq B_H$ almost everywhere, where
\begin{equation}\label{eq:abstract-source-bound}
 B_H=|f(T)|H(u(T),E(T))+|f(0)|H(u(0),E(0))
       +\int_0^T|f_t|H(u,E)\dd t.
\end{equation}
If $B_H\in L^2(\Omega)$, then
$\|Q[d]-Q[u]\|_{L^2(\Omega)}\leq\|B_H\|_{L^2(\Omega)}$.
\end{theorem}
\begin{proof}
Fix a spatial point with absolutely continuous loading paths. Since
$F'(z)=1-z$, the product rule and integration in loading give
\begin{align}
 \partial_t(fF(d))&=f_tF(d)+f(1-d)d_t,\notag\\
 Q[d]&=[fF(d)]_0^T-\int_0^T f_tF(d)\dd t.
 \label{eq:primitive-proof-chain}
\end{align}
Subtract the same identity for $u$ and set $\Delta F=F(d)-F(u)$.
Then
\begin{equation}\label{eq:source-error-identity}
 Q[d]-Q[u]=f(T)\Delta F(T)-f(0)\Delta F(0)
                 -\int_0^T f_t\Delta F\dd t.
\end{equation}
For $\delta=d-u$, the range $d,u\in[0,1]$ yields the two estimates
\begin{align}
 |\Delta F|
 &=\left|\int_u^d(1-z)\dd z\right|\leq|\delta|\leq E,\notag\\
 |\Delta F|
 &=\big|(1-u)\delta-\tfrac12\delta^2\big|
 \leq(1-u)|\delta|+\tfrac12|\delta|^2
 \leq(1-u)E+\tfrac12E^2.
 \label{eq:primitive-difference-chain}
\end{align}
Thus $|\Delta F|\leq H(u,E)$. Applying this bound in
\eqref{eq:source-error-identity} and then integrating in space gives
\begin{align}
 |Q[d]-Q[u]|
 &\leq |f(T)|\,|\Delta F(T)|+|f(0)|\,|\Delta F(0)|
       +\int_0^T|f_t|\,|\Delta F|\dd t\notag\\
 &\leq B_H,\qquad
 \|Q[d]-Q[u]\|_2^2\leq\int_\Omega B_H^2\dd x.
 \label{eq:source-pointwise-chain}
\end{align}
The common-initial-state case has $\Delta F(0)=0$.
\end{proof}

For a practical cellwise estimate, let $u$ interpolate the delivered states
affinely in loading, and suppose $E_{j,K,n}$ enclose separate contributions
to $|d-u|$ on $K\times I_n$. For a reconstructed history, the spatial
contribution includes $|p-u|$; the other contributions may be the loading
and algebraic parts of the history barrier. Define
\begin{equation}\label{eq:lip-factor}
 E_{K,n}=\sum_jE_{j,K,n},\qquad
 L_{K,n}=\min\{1,1-u^{\min}_{K,n}+E_{K,n}/2\}.
\end{equation}
With $V_{K,n}\geq\sup_{x\in K}\int_{I_n}|f_t|\dd t$ and
$f_K^T\geq\sup_K|f(T)|$, the common-initial-state estimate becomes
\begin{equation}\label{eq:source-channel}
 B_{j,K}=f_K^TL_K^TE_{j,K}^T+\sum_nV_{K,n}L_{K,n}E_{j,K,n}.
\end{equation}
Indeed, on $K\times I_n$,
\begin{equation}\label{eq:channel-transfer-chain}
 |F(d)-F(u)|\leq H(u,E_{K,n})
 =E_{K,n}\min\{1,1-u+E_{K,n}/2\}
 \leq L_{K,n}\sum_jE_{j,K,n}.
\end{equation}
For a.e. $x\in K$, the source identity and these bounds give
\begin{align}
 |Q[d]-Q[u]|
 &\leq f_K^TL_K^T\sum_jE_{j,K}^T
   +\sum_n\int_{I_n}|f_t|L_{K,n}\sum_jE_{j,K,n}\dd t\notag\\
 &\leq f_K^TL_K^T\sum_jE_{j,K}^T
   +\sum_nV_{K,n}L_{K,n}\sum_jE_{j,K,n}\notag\\
 &=\sum_j\left(f_K^TL_K^TE_{j,K}^T
       +\sum_nV_{K,n}L_{K,n}E_{j,K,n}\right)
 =\sum_jB_{j,K}.
 \label{eq:source-channel-summation}
\end{align}
The variation bounds $V_{K,n}$ cover the full prescribed intervals.

The remaining terms compare $Q[u]$ with the endpoint accumulation
$\widetilde Q_{\raw}$ and its stored coefficient field $Q_{\raw}$.
On $I_n$, put $\Delta u_n=u_n-u_{n-1}$, $t=t_{n-1}+\theta\Delta t_n$,
$f_{0,n}=\int_0^1f(\theta)\dd\theta$, and
$f_{1,n}=\int_0^1\theta f(\theta)\dd\theta$. At a fixed spatial point,
\begin{align}
 Q_n[u]
 &=\int_{I_n}f(t)(1-u(t))u_t(t)\dd t\notag\\
 &=\int_0^1 f(\theta)(1-u_{n-1}-\theta\Delta u_n)
                 \frac{\Delta u_n}{\Delta t_n}\Delta t_n\dd\theta\notag\\
 &=\Delta u_n\big[f_{0,n}(1-u_{n-1})-f_{1,n}\Delta u_n\big],
 \label{eq:interval-source-chain}\\
 \widetilde Q_{\raw,n}
 &=f_n(1-u_n)\Delta u_n
   =\Delta u_n\big[f_n(1-u_{n-1})-f_n\Delta u_n\big].\notag
\end{align}
Subtracting these expressions gives the exact endpoint-rule defect
\begin{equation}\label{eq:defect}
 D_n=\widetilde Q_{\raw,n}-Q_n[u]
 =\Delta u_n\big[(f_n-f_{0,n})(1-u_{n-1})-(f_n-f_{1,n})\Delta u_n\big].
\end{equation}
The polynomial realizations enclose these exact loading moments by
outward arithmetic.

Summing the signed defects before bounding them gives
$|\widetilde Q_{\raw}-Q[u]|=|\sum_nD_n|\leq B_{q,K}$ on a cell.
The coefficient enclosures in Section~\ref{sec:verified-arithmetic}
supply the arithmetic bound
$|Q_{\raw}-\widetilde Q_{\raw}|\leq B_{\mathrm{fp},K}$.

\begin{theorem}[Certificate for the delivered field]\label{thm:source}
Suppose the history, variation, endpoint-defect and arithmetic bounds above
are rigorous on a measurable spatial partition. Define
\begin{equation}\label{eq:certificate}
 B_K=\sum_jB_{j,K}+B_{q,K}+B_{\mathrm{fp},K},\qquad
 C=\left(\sum_K|K|B_K^2\right)^{1/2}.
\end{equation}
Then $|Q_{\raw}-Q[d]|\leq B_K$ on each cell and
$\|Q_{\raw}-Q[d]\|_2\leq C$. If $q_-\leq\|Q_{\raw}\|_2$ and
$q_->C$, then
\begin{equation}\label{eq:relative-certificate}
 \frac{\|Q_{\raw}-Q[d]\|_2}{\|Q[d]\|_2}\leq\frac{C}{q_--C}.
\end{equation}
\end{theorem}
\begin{proof}
Decompose the delivered error using the interpolated source and the
exact endpoint polynomial. On each cell,
\begin{align}
 Q_{\raw}-Q[d]
 &=(Q_{\raw}-\widetilde Q_{\raw})
   +(\widetilde Q_{\raw}-Q[u])+(Q[u]-Q[d]),\notag\\
 |Q_{\raw}-Q[d]|
 &\leq B_{\mathrm{fp},K}+B_{q,K}+\sum_jB_{j,K}=B_K.
 \label{eq:delivered-error-chain}
\end{align}
The last step uses coefficient enclosure, \eqref{eq:defect}, and
\eqref{eq:source-channel-summation}. Hence
\begin{equation}\label{eq:absolute-norm-chain}
 \|Q_{\raw}-Q[d]\|_2^2
 =\sum_K\int_K|Q_{\raw}-Q[d]|^2\dd x
 \leq\sum_K|K|B_K^2=C^2.
\end{equation}
The reverse triangle inequality gives
\begin{align}
 \|Q[d]\|_2
 &\geq\|Q_{\raw}\|_2-\|Q_{\raw}-Q[d]\|_2
 \geq q_--C>0,\notag\\
 \frac{\|Q_{\raw}-Q[d]\|_2}{\|Q[d]\|_2}
 &\leq\frac{C}{q_--C}.
 \label{eq:relative-norm-chain}
\end{align}
Outward evaluation rounds $C$ upward and the positive denominator
downward, preserving the bound.
\end{proof}

For $q_-\leq C$, adaptation continues using the absolute bound.
Section~\ref{sec:algorithm} gives the resulting acceptance procedure.
\subsection{Endpoint consistency and decay of the certificate}
\label{snum:consistency}
The following bound on the endpoint defect yields sufficient conditions
for decay of the certificate.

\begin{proposition}[Consistency of endpoint accumulation]\label{snum:prop-endpoint}
Let $f$ satisfy the assumptions of Section~\ref{sec:model}, put
$F_* =\|f\|_{L^\infty(\Omega\times(0,T))}$, and let $u$ be the
nondecreasing affine interpolant of fields in $[0,1]$. Set
$\tau=\max_n\Delta t_n$ and $\eta=\max_n\|\Delta u_n\|_\infty$.
For the exact defects in \eqref{eq:defect},
\begin{equation}\label{snum:defect-consistency}
 \left\|\sum_nD_n\right\|_\infty
 \leq\frac{L\tau+F_*\eta}{2}.
\end{equation}
If $\|d-u\|_{L^\infty(\Omega\times[0,T])}\leq\delta$, then
$\eta\leq L\tau/a+2\delta$.
\end{proposition}
\begin{proof}
At a fixed spatial point, use $u_{n-1}=u_n-\Delta u_n$ and
$f_{0,n}-f_{1,n}=\int_0^1(1-\theta)f(\theta)\dd\theta$ in
\eqref{eq:defect}. Then
\begin{align}
 D_n&=\Delta u_n\left[(f_n-f_{0,n})(1-u_n)
       -\Delta u_n\int_0^1(1-\theta)f(\theta)\dd\theta\right],\notag\\
 |f_n-f_{0,n}|
 &\leq\int_0^1L\Delta t_n(1-\theta)\dd\theta=L\Delta t_n/2,\notag\\
 |D_n|&\leq\tfrac12L\Delta t_n\Delta u_n
                +\tfrac12F_*(\Delta u_n)^2.
 \label{snum:endpoint-chain}
\end{align}
Here $0\leq1-u_n\leq1$ and $\Delta u_n\geq0$. Consequently,
\begin{align*}
 \left|\sum_nD_n\right|
 &\leq\tfrac12L\tau\sum_n\Delta u_n
       +\tfrac12F_*\eta\sum_n\Delta u_n\\
 &=\tfrac12(L\tau+F_*\eta)(u_N-u_0)
 \leq\tfrac12(L\tau+F_*\eta).
\end{align*}
Taking the essential supremum proves \eqref{snum:defect-consistency}.
Proposition~\ref{prop:history-limit} also yields
\begin{align*}
 \|u_n-u_{n-1}\|_\infty
 &\leq\|u_n-d(t_n)\|_\infty
       +\|d(t_n)-d(t_{n-1})\|_\infty
       +\|d(t_{n-1})-u_{n-1}\|_\infty\\
 &\leq2\delta+L\Delta t_n/a\leq2\delta+L\tau/a.
\end{align*}
Maximizing over $n$ proves the bound for $\eta$.
\end{proof}

\begin{corollary}[Sufficient conditions for a vanishing certificate]
\label{snum:cor-decay}
Consider histories with loading size $\tau_k\to0$ and rigorous history
channel bounds satisfying
$\delta_k=\max_{K,n}\sum_jE_{j,K,n}\to0$, including the terminal bounds.
Assume the computed variation bounds obey
$\max_K(f_K^T+\sum_nV_{K,n})\leq W_*$ uniformly, and
$\epsilon_k=\max_K B_{\mathrm{fp},K}\to0$. In \eqref{eq:certificate},
one may replace $B_{q,K}$ by the smaller of its original value and a
rigorous evaluation of $(L\tau_k+F_*\eta_k)/2$. In exact arithmetic the
resulting bound $C_k^*$ satisfies
\begin{equation}\label{snum:certificate-decay}
 C_k^*\leq |\Omega|^{1/2}
 \left[(W_*+F_*)\delta_k+
       \frac{L}{2}\left(1+\frac{F_*}{a}\right)\tau_k+\epsilon_k\right]
 \longrightarrow0.
\end{equation}
The same conclusion holds for outward evaluation if its excess enclosure
tends to zero. If $\|Q[d]\|_2>0$ and
$\|Q_{\raw,k}\|_2-q_{-,k}\to0$, the relative certificate tends to zero.
\end{corollary}
\begin{proof}
The smaller of the two valid defect bounds preserves
Theorem~\ref{thm:source}. Since $L_{K,n}\leq1$, the history contribution
is bounded by $W_*\delta_k$. Proposition~\ref{snum:prop-endpoint} gives
\begin{align*}
 \eta_k&\leq L\tau_k/a+2\delta_k,\\
 B_{q,K}^*&\leq\tfrac12L\tau_k+\tfrac12F_*\eta_k
 \leq F_*\delta_k+\tfrac12L(1+F_*/a)\tau_k,\\
 B_K^*&\leq W_*\delta_k+B_{q,K}^*+\epsilon_k
 \leq(W_*+F_*)\delta_k+\tfrac12L(1+F_*/a)\tau_k+\epsilon_k,\\
 C_k^*&=\left(\sum_K|K|(B_K^*)^2\right)^{1/2}
 \leq|\Omega|^{1/2}\max_K B_K^*\longrightarrow0.
\end{align*}
This proves \eqref{snum:certificate-decay}. For outward evaluation write
$\widehat C_k=C_k^*+\omega_k$, with $\omega_k\geq0$ and $\omega_k\to0$;
in exact arithmetic take $\omega_k=0$. Then $\widehat C_k\to0$.
Put $q=\|Q[d]\|_2>0$ and
$\sigma_k=\|Q_{\raw,k}\|_2-q_{-,k}\to0$. Reliability implies
\begin{align*}
 q_{-,k}-\widehat C_k
 &=\|Q_{\raw,k}\|_2-\sigma_k-\widehat C_k
 \geq q-\sigma_k-2\widehat C_k,\\
 0\leq\frac{\widehat C_k}{q_{-,k}-\widehat C_k}
 &\leq\frac{\widehat C_k}{q-\sigma_k-2\widehat C_k}\longrightarrow0,
\end{align*}
where both denominators are positive for all sufficiently large $k$.
\end{proof}

For the constant provider, vanishing residual suprema and reconstruction
defects imply $\delta_k\to0$. Applying the corollary to a generated
adaptive sequence would require these approximation properties and
control of outward-enclosure excess, with increasing precision if needed.
The experiments use the signed polynomial defect bound.

\section{Fully Computable Realizations}\label{sec:realizations}
We construct admissible histories, residual majorants, and norm bounds on
interval and tensor-product meshes. All experiments start from $d_0=0$.
General initial data would require an admissible initial reconstruction
and its enclosure.

\subsection{One-dimensional fully certified realization}\label{sec:reconstruction}
Let $\Omega=(0,1)$. We reconstruct the finite-element history by a
nondecreasing cubic spline.

The numerical states use continuous piecewise linear finite elements with
lumped reaction and loading terms. For nodal masses $m_i$ and the stiffness
matrix $K_b$, the free-node equation at load $t_n$ is
\begin{equation}\label{eq:lumped}
 \left[K_b+\operatorname{diag}\{m_i(a+f_n(x_i))\}\right]\mathbf u_n
 =\{m_i f_n(x_i)\}_i,\qquad\mathbf u_n\geq\mathbf u_{n-1}.
\end{equation}
The matrix property can be seen directly. If $h_i=x_{i+1}-x_i$, then
\begin{equation}\label{eq:discrete-coercivity}
 \mathbf z^TA_h\mathbf z
 =\sum_i\frac b{h_i}(z_{i+1}-z_i)^2
      +\sum_i m_i(a+f_n(x_i))z_i^2>0\qquad(\mathbf z\ne0).
\end{equation}
Moreover, $(A_h)_{ij}\leq0$ for $i\ne j$ and the row sum is
$m_i(a+f_n(x_i))>0$. If $A_h\mathbf z\geq0$ but $z_i<0$ is a
minimal component, then
\[
 (A_h\mathbf z)_i
 =\sum_{j\ne i}(A_h)_{ij}(z_j-z_i)
       +m_i(a+f_n(x_i))z_i<0,
\]
which is a contradiction. Thus the inverse is nonnegative, proving the
nonsingular M-matrix property. A principal free-node matrix retains
positive definiteness and has an even larger row sum after negative
off-diagonal entries are removed; the same argument applies to it.
A monotone active-set release procedure enforces complementarity
\cite{Hintermuller2002,Heister2015}, using restricted conjugate gradients
or banded direct solves. The signed projected residual defined below
records the algebraic error.

Use a nonnegative open clamped cubic B-spline basis $\{B_i\}$ and
affine loading interpolation. The following local polar recovery
\cite{deBoor2001,Farouki2012}, followed by coefficient constraints,
provides the required range, irreversibility, and zero boundary flux.
\label{app:recovery}
For each basis coefficient, select four distinct neighboring nodes
$x_0,\ldots,x_3$ and their Lagrange polynomials
$\ell_j(x)=\prod_{m\ne j}(x-x_m)/(x_j-x_m)$. The local interpolating
cubic is $v_i^n(x)=\sum_{j=0}^3u_n(x_j)\ell_j(x)$. Its B-spline
coefficient is obtained by applying the cubic polar functional to this
polynomial. Let the three interior knots of a cubic B-spline basis function
be $\xi_1,\xi_2,\xi_3$. The polar-functional values of the
monomials $1,x,x^2,x^3$ are
\[
 1,\quad
 \frac{\xi_1+\xi_2+\xi_3}{3},\quad
 \frac{\xi_1\xi_2+\xi_1\xi_3+\xi_2\xi_3}{3},\quad
 \xi_1\xi_2\xi_3 .
\]
Denote this functional by $\mathcal P_i$. Linearity gives an explicit
four-point recovery formula,
\[
 \widetilde c_i^n=\mathcal P_i[v_i^n]
 =\sum_{j=0}^3w_{ij}u_n(x_j),\qquad w_{ij}=\mathcal P_i[\ell_j].
\]
The mesh-dependent weights reproduce cubics before limiting and may have
either sign. Impose admissibility by setting
\begin{equation}\label{eq:coeff-limit}
 c_i^n=\min\{1,\max(c_i^{n-1},\widetilde c_i^n,0)\},\qquad
 p_n=\sum_i c_i^n B_i .
\end{equation}
Boundary coefficient pairs are identified before limiting. If $I$ is the
last coefficient index, differentiation of the clamped basis gives
\begin{equation}\label{eq:1d-neumann-coefficients}
 p_n'(0)=3(c_1^n-c_0^n)/h_{\rm left}=0,\qquad
 p_n'(1)=3(c_I^n-c_{I-1}^n)/h_{\rm right}=0.
\end{equation}
The same candidate and previous value are used for both members of a
boundary pair, so limiting preserves these identities at every load. The nonnegative basis forms a partition of unity, so
$0\leq c_i^n\leq1$ implies $0\leq p_n\leq1$. Affine interpolation between
successive reconstructed states gives the explicit order calculation
\[
 p_t(x,t)=\frac{p_n(x)-p_{n-1}(x)}{t_n-t_{n-1}}
 =\sum_i\frac{c_i^n-c_i^{n-1}}{t_n-t_{n-1}}B_i(x)\geq0,
 \qquad t\in I_n.
\]
Unchanged nodal patches give unchanged candidate coefficients and hence
preserve inactivity. The stored coefficients define $p$; rational knot
insertion encloses its Bernstein representation. Recovery, limiting, and
coefficient rounding are covered by the discrepancy from the delivered
P1 field. Since $p$ and $u$ are affine in loading,
\begin{equation}\label{eq:reconstruction-defect-chain}
 |p(t)-u(t)|
 \leq(1-\theta)|p_{n-1}-u_{n-1}|+\theta|p_n-u_n|
 \leq\max\{r_{K,n-1}^{\rm rec},r_{K,n}^{\rm rec}\}
 \quad\text{on }K\times I_n.
\end{equation}
The endpoint differences are bounded in a common Bernstein representation.

On a cell whose four supporting coefficients are unchanged, set
$\lambda=(\LL_t p-f)_+$. On an evolving cell set $\lambda=0$.
Thus $R=f-\LL_t p$ gives $r=R_+$ on frozen cells and $r=R$ on
evolving cells, satisfying \eqref{eq:complementary-residual}.
This contact-sensitive construction follows the treatment of admissible
constraint reactions in \cite{Nochetto2003,Walloth2022}.

Let $a_n(x)$ be the piecewise linear interpolant of the signed projected
nodal residual of \eqref{eq:lumped}, divided by nodal mass. At a node with
strictly positive accepted damage increment the signed residual is retained.
At a frozen node its negative part is removed. Define
\[
 R_n=f_n-\LL_{t_n}p_n,\qquad R_{x,n}=R_n-a_n.
 \]
For $\theta=(t-t_{n-1})/(t_n-t_{n-1})$, write
$\mathcal I_n$ for affine interpolation between endpoints. The exact identity
\begin{align}
 R(t)&=\mathcal I_n R_x+\mathcal I_n a+R_{\tau,n}(t),\label{eq:res-split}\\
 R_{\tau,n}(t)
 &=f(t)(1-p(t))-\mathcal I_n\!\left[f(1-p)\right](t)
 \label{eq:time-residual}
\end{align}
follows because $b p_{xx}-a p$ is affine on $I_n$,
\[
 R(t)-\mathcal I_nR
 =f(t)(1-p(t))-\mathcal I_n[f(1-p)](t)=R_{\tau,n}(t).
\]
The definition $R_{x,n}=R_n-a_n$ retains rounding in the stored
algebraic field.

For a general loading clock, the interval term can also be written as
\begin{equation}\label{snum:loading-residual}
 R_{\tau,n}=(f-\mathcal I_nf)(1-p)
          +\theta(1-\theta)\Delta f_n\Delta p_n.
\end{equation}
Indeed, affine interpolation of $p$ gives
\begin{align*}
 R_{\tau,n}
 &=(f-\mathcal I_nf)(1-p)
   +(\mathcal I_nf)(1-p)-\mathcal I_n[f(1-p)]\\
 &=(f-\mathcal I_nf)(1-p)
   +\theta(1-\theta)\Delta f_n\Delta p_n.
\end{align*}

On an evolving cell, absolute endpoint bounds for $R_x$ and $a$, together
with an interval bound for $R_{\tau,n}$, supply
$\rho_x,\rho_a,\rho_\tau$. On a frozen cell use
\begin{equation}\label{eq:positive-split}
 (A+B+C)_+\leq A_++|B|+|C|.
\end{equation}
Bernstein bounds of the positive spatial part on frozen cells and the
absolute parts elsewhere give the full-interval envelopes. In particular,
\eqref{eq:time-residual} detects an interior drive when both endpoints
vanish.

The resulting nonnegative envelopes satisfy
\begin{equation}\label{snum:rho}
 |r|\leq\rho_{x,K,n}+\rho_{\tau,K,n}+\rho_{a,K,n}
 \quad\text{on }K\times I_n.
\end{equation}
For each channel $j\in\{x,\tau,a\}$, define the prefix residual and
a companion constant radius by
\begin{equation}\label{eq:prefix}
 \overline\rho_j^n|_K=\max_{m\leq n}\rho_{j,K,m},
 \qquad
 \delta_j^n=\max_{K,m\leq n}
 \frac{\rho_{j,K,m}}{\inf_{K\times I_m}(a+f)} .
\end{equation}
Let $\LL_0=-b\partial_{xx}+a$ with the same Neumann boundary condition.
Its inverse is positive. Both
$\LL_0^{-1}\overline\rho_j^n$ and the constant $\delta_j^n$ are
supersolutions for the channel residual on $I_n$. Consequently,
\begin{equation}\label{eq:causal-barrier}
 e_{j,n}=\min\{\LL_0^{-1}\overline\rho_j^n,\delta_j^n\},
 \qquad e_n=e_{x,n}+e_{\tau,n}+e_{a,n},
\end{equation}
satisfies Theorem~\ref{thm:history}.
The supersolution property follows in three steps. First let
$z_n=\LL_0^{-1}\overline\rho_j^n$. The weak equation
and the negative-part test $z_n^-\!=\max(-z_n,0)$ give
\[
 -\|z_n^-\|_0^2=A_0(z_n,z_n^-)
 =(\overline\rho_j^n,z_n^-)\geq0,
 \qquad z_n\geq0.
\]
Consequently, for $t\in I_n$,
\begin{equation}\label{eq:green-supersolution-chain}
 \LL_tz_n=\overline\rho_j^n+f(t)z_n
 \geq\overline\rho_j^n\geq\rho_j(t).
\end{equation}
Applying the same negative-part argument to
$\LL_0(z_n-z_{n-1})=\overline\rho_j^n-\overline\rho_j^{n-1}\geq0$
gives $z_n\geq z_{n-1}$. Second, the constant in Eq.~\eqref{eq:prefix}
satisfies, on every $K\times I_n$,
\begin{equation}\label{eq:local-constant-chain}
 (a+f(x,t))\delta_j^n
 \geq \inf_{K\times I_n}(a+f)\,
       \frac{\rho_{j,K,n}}{\inf_{K\times I_n}(a+f)}
 =\rho_{j,K,n}.
\end{equation}
The prefix maximum makes this constant nondecreasing as well.

Third, the minimum of two weak supersolutions for the same right-hand
side remains a weak supersolution. Let $v,w$ be weak supersolutions for
$\rho\geq0$, and let $\chi_\delta$ be a nondecreasing Lipschitz
approximation to $\mathbf1_{(0,\infty)}$. For nonnegative bounded $\phi\in H^1$, use
$\phi\chi_\delta(w-v)$ in the inequality for $v$ and
$\phi[1-\chi_\delta(w-v)]$ in that for $w$. Both are admissible $H^1$
tests. Adding and expanding their gradients gives
\begin{align}
 &b\int_\Omega[\chi_\delta\nabla v+(1-\chi_\delta)\nabla w]
                         \cdot\nabla\phi\dd x
 +\int_\Omega(a+f)[\chi_\delta v+(1-\chi_\delta)w]\phi\dd x\notag\\
 &\hspace{1cm}\geq(\rho,\phi)
       +b\int_\Omega\phi\chi_\delta'(w-v)|\nabla(w-v)|^2\dd x
 \geq(\rho,\phi).\label{eq:min-super-chain}
\end{align}
Here $\chi_\delta$ is evaluated at $w-v$ in every occurrence. Letting
$\delta\downarrow0$, dominated convergence in the terms on the left
gives $A_f(\min(v,w),\phi)\geq(\rho,\phi)$. On $\{v=w\}$ their
gradients agree almost everywhere, so the chosen value of the limiting
indicator there is immaterial. Truncation extends the conclusion to
unbounded nonnegative $H^1$ tests. Thus the minimum is a weak
supersolution. Applying this result to $z_n$ and $\delta_j^n$, and then
adding the three channel inequalities, proves the residual requirement in
Theorem~\ref{thm:history}. Finally,
$z_n\geq z_{n-1}$ and $\delta_j^n\geq\delta_j^{n-1}$ imply
$\min(z_n,\delta_j^n)\geq\min(z_{n-1},\delta_j^{n-1})$ pointwise.
This proves the required causal order as well as the supersolution property.

Computed cell maxima enclose the $H^1$ supersolutions in
\eqref{eq:causal-barrier}. For $a=1$ and $b=k^{-2}$, the Neumann kernel of
$\LL_0=-k^{-2}\partial_{xx}+1$ is
\begin{equation}\label{eq:green}
 G(x,y)=
 \frac{
 e^{-k|x-y|}+e^{-k(x+y)}
 +e^{-k(2-x-y)}+e^{-k(2-|x-y|)}
 }{(2/k)(1-e^{-2k})}.
\end{equation}
An equivalent hyperbolic representation follows by expanding the
exponentials. Writing $x_< =\min(x,y)$ and $x_> =\max(x,y)$,
Eq.~\eqref{eq:green} becomes
\[
 G(x,y)=\frac{k\cosh(kx_<)\cosh(k(1-x_>))}{\sinh k}.
\]
Away from $x=y$, differentiation gives $G_{xx}=k^2G$. At the two domain
ends, the factor $\sinh(0)$ gives $G_x(0,y)=G_x(1,y)=0$. The two
expressions agree at $x=y$, and their derivative jump is
\begin{align}
 G_x(y^+,y)-G_x(y^-,y)
 &=-\frac{k^2}{\sinh k}
   [\cosh(ky)\sinh(k(1-y))\notag\\
 &\qquad{}+\sinh(ky)\cosh(k(1-y))]\notag\\
 &=-k^2,\quad
 -k^{-2}[G_x(y^+,y)-G_x(y^-,y)]=1.\label{eq:green-jump-chain}
\end{align}
Cellwise integration by parts consequently gives
$\LL_0G(\cdot,y)=\delta_y$ in the weak sense for interior $y$;
endpoint values follow by limits, and their contributions to the response
integral vanish. Coercivity identifies this response with the unique Neumann
inverse. Every term in the numerator is positive, which permits termwise upper
bounds for the response. For constant forcing $r_J\geq0$ on each cell $J$,
\[
 (\LL_0^{-1}r)(x)=\sum_Jr_J\int_JG(x,y)\dd y
 \leq\sum_J\overline G_{KJ}r_J,\qquad x\in K,
\]
where $\overline G_{KJ}\geq |J|\sup_{K\times J}G$.
Each exponential is bounded by minimizing its exponent over $K\times J$.
For the dyadic geometries and power-of-two $k$, these arguments are
exactly representable.

The exponential factors are enclosed directly by a Taylor-based
construction. For $z\geq0$, define
$P_{12}(z/128)=\sum_{j=0}^{12}(z/128)^j/j!$. Positivity of the omitted
Taylor terms gives the chain
\[
 0<P_{12}(z/128)\leq e^{z/128},\quad
 e^{-z/128}\leq P_{12}(z/128)^{-1},\quad
 e^{-z}\leq\big[P_{12}(z/128)^{-1}\big]^{128}.
\]
The polynomial is evaluated downward before taking an upward reciprocal.
Seven repeated squarings produce an upper enclosure for
$\exp(-z)$. All operations are outward-rounded, following the inclusion principles
of interval arithmetic \cite{Moore2009}.
If $U_{2k}\geq e^{-2k}$ is the resulting upper exponential bound,
then $1-e^{-2k}\geq1-U_{2k}$. Division uses an upward numerator and a positive downward endpoint of
$1-U_{2k}$.

For source transfer set $E_{x,K,n}=\widehat e_{x,K,n}+r^{\mathrm{rec}}_{K,n}$,
$E_{\tau,K,n}=\widehat e_{\tau,K,n}$, and
$E_{a,K,n}=\widehat e_{a,K,n}$, where
$r^{\mathrm{rec}}_{K,n}\geq\sup_{K\times I_n}|p-u|$ is obtained from the
endpoint Bernstein differences. Equations~\eqref{eq:source-channel}--
\eqref{eq:relative-certificate} then give the complete source bound.

\subsection{A controlled two-dimensional certificate}\label{sec:2d-certificate}

On $\Omega=(0,1)^2$, use a structured tensor grid, homogeneous Neumann
conditions, a polynomial drive, and the provider in
Proposition~\ref{prop:constant}. Section~\ref{sec:sen} uses the same tensor
construction on a $1\times2$ rectangle with a mechanically generated drive
and verified auxiliary providers.

Let $B_i(x)$ and $B_j(y)$ be the open clamped cubic bases on the two
coordinate partitions, with simple interior knots, and define
\begin{equation}\label{eq:tensor-recovery}
 p_n(x,y)=\sum_{i,j}c^n_{ij}B_i(x)B_j(y),\qquad
 0\leq c^{n-1}_{ij}\leq c^n_{ij}\leq1.
\end{equation}
Apply the polar recovery in each coordinate. Identify the first two and
last two coefficient rows or columns at each boundary, then apply the
coefficient bounds to these groups, starting from zero. The following
proposition verifies that this order of operations preserves admissibility.

\begin{proposition}[Admissibility of the tensor reconstruction]\label{prop:tensor}
The reconstruction in Eq.~\eqref{eq:tensor-recovery} is globally $C^2$
in space, lies in $[0,1]$, has zero normal derivative on the boundary,
and is nondecreasing throughout each space--loading cell after affine
loading interpolation. Its distributional Laplacian equals its cellwise
Laplacian.
\end{proposition}
\begin{proof}
\emph{Range, continuity and loading order.}
A cubic B-spline with simple interior knots has continuity $C^{3-1}=C^2$.
The tensor product preserves this continuity in both coordinates.
At each point, $B_i(x)B_j(y)\geq0$ and
$\sum_{i,j}B_i(x)B_j(y)=(\sum_iB_i(x))(\sum_jB_j(y))=1$. Therefore
\begin{align}
 0&\leq p_n(x,y)=\sum_{i,j}c^n_{ij}B_i(x)B_j(y)
                  \leq\sum_{i,j}B_i(x)B_j(y)=1,\notag\\
 p_n-p_{n-1}&=\sum_{i,j}(c^n_{ij}-c^{n-1}_{ij})B_i(x)B_j(y)\geq0,\notag\\
 p_t&=\frac1{\Delta t_n}\sum_{i,j}(c^n_{ij}-c^{n-1}_{ij})B_i(x)B_j(y)
       \geq0.\label{eq:tensor-order-chain}
\end{align}
All statements hold in cell interiors as well as at nodes. Knot insertion
forms convex combinations of coefficients; inserting enough knots to
obtain cellwise Bernstein form consequently preserves the known range
$[0,1]$. Intersecting extracted coefficient intervals with this range is
therefore valid.

\emph{The four Neumann edges and their corners.}
For a clamped cubic on a uniform coordinate mesh with width $h_x$,
the left-edge derivative is determined by the first two spline coefficients.
The basis derivative formula gives
\begin{align}
 \partial_xp_n(0,y)
 &=\frac3{h_x}\sum_j(c^n_{1j}-c^n_{0j})B_j(y)=0,\notag\\
 \partial_xp_n(1,y)
 &=\frac3{h_x}\sum_j(c^n_{I j}-c^n_{I-1,j})B_j(y)=0.
 \label{eq:tensor-neumann-chain}
\end{align}
Here $I$ is the last coefficient index in the $x$ direction. The analogous
formulas with adjacent columns give $\partial_yp_n(x,0)=\partial_yp_n(x,1)=0$.
For nonuniform coordinate partitions the respective boundary cell widths
replace $h_x,h_y$. The corner constraints simply identify the four
coefficients in each corner group; row and column identifications commute.
Starting from identical coefficients in each group, applying the same
clamp to their common candidate and their common previous value preserves
both these equalities and coefficient order. Affine loading interpolation
preserves the zero boundary derivatives.

\emph{Absence of interface residual distributions.}
Let $\phi\in H^1(\Omega)$ and integrate over the rectangles separately.
The piecewise classical Laplacian satisfies
\begin{align}
 \sum_K\int_K(-\Delta p)\phi\dd x
 &=\sum_K\int_K\nabla p\cdot\nabla\phi\dd x
   -\sum_K\int_{\partial K}(\nabla p\cdot n_K)\phi\dd s\notag\\
 &=\int_\Omega\nabla p\cdot\nabla\phi\dd x.
 \label{eq:tensor-no-jump-chain}
\end{align}
The two normals at an interior edge have opposite signs and $\nabla p$
has the same trace on both sides, so its two edge terms cancel. The
outer-edge terms vanish by Eq.~\eqref{eq:tensor-neumann-chain}.
The identity extends from smooth tests by density. It proves that the
cellwise Laplacian is the weak Neumann Laplacian, with all interface
terms canceled. The residual can therefore be evaluated in strong
form on cells and used in the weak comparison theorem.
\end{proof}

The state $u_n$ is bilinear. A conservative five-point discretization
uses half control volumes at edges and quarter volumes at corners.
With positive diagonal nodal-mass matrices $M_x,M_y$, its mass-weighted
matrix is
\[
 A_h=K_x\otimes M_y+M_x\otimes K_y
       +\operatorname{diag}\{m_i^xm_j^y(a+f_{ij})\}.
\]
The stiffness terms are positive semidefinite with nonpositive
off-diagonals. The reaction is positive definite, and the row sum
$m_i^xm_j^y(a+f_{ij})$ is positive. The discrete minimum argument in
Section~\ref{sec:reconstruction} therefore gives $A_h^{-1}\geq0$ and
uniqueness. CG uses a cosine-transform inverse of a constant-reaction
Neumann operator, with relative and absolute tolerances
$2\times10^{-13}$ and $10^{-15}$. States are clamped to $[u_{n-1},1]$,
and the residual is evaluated on these stored states.

The residual separation extends by replacing $p_{xx}$ by $\Delta p$.
For the stored bilinear algebraic-residual field $a_n$, define
$R_n=f_n-\LL_{t_n}p_n$ and $R_{x,n}=R_n-a_n$. Affine interpolation
commutes with $b\Delta p-ap$, giving exactly
\begin{equation}\label{eq:2d-residual-split}
 R(t)=\mathcal I_nR_x+\mathcal I_na+
       f(t)(1-p(t))-\mathcal I_n[f(1-p)](t).
\end{equation}
On a rectangle whose sixteen supporting coefficients are unchanged,
one may set $\lambda=(-R)_+$ and use $r=R_+$; on other rectangles
set $\lambda=0$ and use $r=R$. The positive-part inequality
\eqref{eq:positive-split} proves the same three-channel majorization
as in one dimension. The reported monotone experiment takes $\lambda=0$
everywhere and bounds the full residual, which is also complementary
where the reconstruction happens to be frozen.

For the piecewise affine clocks actually implemented in this experiment,
the interval residual has the exact form
\begin{equation}\label{eq:2d-loading-residual}
 R_{\tau,n}(\theta)=\theta(1-\theta)\Delta f_n\Delta p_n,
 \qquad |R_{\tau,n}|\leq\tfrac14|\Delta f_n\Delta p_n|.
\end{equation}
To see the cancellation explicitly, write
$f=f_{n-1}+\theta\Delta f_n$ and $p=p_{n-1}+\theta\Delta p_n$. Then
\begin{align}
 f(1-p)&=f_{n-1}(1-p_{n-1})-\theta^2\Delta f_n\Delta p_n\notag\\
 &\quad+\theta[\Delta f_n(1-p_{n-1})-f_{n-1}\Delta p_n],\notag\\
 \mathcal I_n[f(1-p)]&=f_{n-1}(1-p_{n-1})-\theta\Delta f_n\Delta p_n\notag\\
 &\quad+\theta[\Delta f_n(1-p_{n-1})-f_{n-1}\Delta p_n],\notag\\
 f(1-p)-\mathcal I_n[f(1-p)]&=(\theta-\theta^2)\Delta f_n\Delta p_n.
 \label{eq:tensor-loading-chain}
\end{align}
The maximum $\max_{[0,1]}\theta(1-\theta)=1/4$ gives the interval bound.
For higher-degree clocks, \eqref{snum:loading-residual} also retains
the deviation from affine interpolation, as used for the SEN amplitudes
in Section~\ref{sec:sen}.

Tensor Bernstein bounds enclose $R_{x,n}$, $a_n$,
$\Delta f_n\Delta p_n$, and $p_n-u_n$ on each rectangle. Spatial maxima
give $\rho_{j,n}$ in \eqref{eq:constant-barrier}, with $a_0=1$.
Equations~\eqref{eq:source-channel}--\eqref{eq:relative-certificate}
then give the source certificate. Section~\ref{sec:verified-arithmetic}
details the coefficient conversion, tensor contractions, and norm
integration used in this evaluation.

For monotone drives, Proposition~\ref{prop:2d-endpoint} identifies the
obstacle history with the stationary solution at each time. The reference
uses a separate fine tensor grid, Neumann cubic interpolation, and
continuous polynomial source integration. Its stationary residual gives
an enclosure at each time, as follows.

At a fixed time, write $v=p_{\refe}$ and
$r_{\refe}=f-\LL_tv$, with $|r_{\refe}|\leq\rho(t)$.
Set $\delta=\rho(t)/a$. For the exact stationary solution $z$,
\begin{align}
 \LL_t(v+\delta)-f&=-r_{\refe}+(a+f)\delta\geq0,\notag\\
 \LL_t(v-\delta)-f&=-r_{\refe}-(a+f)\delta\leq0.\label{eq:reference-barrier-chain}
\end{align}
Testing the first comparison with $(z-v-\delta)_+$ and the second
with $(v-\delta-z)_+$ gives zero for both positive parts by coercivity.
Thus $|z-v|\leq\delta$. The stationary comparison applies separately at
each time to the interpolatory reconstruction $v$. The bilinear reference
state interpolation lies in $[0,1]$; combining $|d-u_{\refe}|\leq\delta+|v-u_{\refe}|$ with
Theorem~\ref{thm:transfer} encloses its continuous accumulated source.
Production and reference share the interval-arithmetic primitives in
Section~\ref{sec:verified-arithmetic}.
\subsection{Verified auxiliary supersolutions}\label{sec:auxiliary-providers}
For spatially varying supersolutions, let
$\overline\rho_j^n|_K=\max_{m\leq n}\rho_{j,K,m}$ and
$\LL_0=-b\Delta+a$. The exact weak Neumann response
\begin{equation}\label{snum:elliptic-provider}
 \LL_0 z_{j,n}=\overline\rho_j^n
\end{equation}
is ordered by the negative-part tests. Set $z_{j,0}=0$ and $\overline\rho_j^0=0$. With
$\xi=(-z_{j,n})_+$ and $\eta=(z_{j,n-1}-z_{j,n})_+$,
\begin{align*}
 \|\xi\|_0^2
 &=-A_0(z_{j,n},\xi)=-(\overline\rho_j^n,\xi)\leq0,\\
 \|\eta\|_0^2
 &=-A_0(z_{j,n}-z_{j,n-1},\eta)
 =-(\overline\rho_j^n-\overline\rho_j^{n-1},\eta)\leq0,\\
 z_{j,n}&\geq z_{j,n-1}\geq0,\qquad
 \LL_tz_{j,n}=\overline\rho_j^n+fz_{j,n}\geq\rho_j(t).
\end{align*}
Section~\ref{sec:reconstruction} evaluates this response through the
one-dimensional Green kernel. In two dimensions we use a verified
auxiliary solve.

Choose a continuous bilinear $g_j$ whose vertex values are the maxima of the adjacent cell
prefix residuals. Convex interpolation gives $g_j\geq\overline\rho_j$
on each cell. Let $\LL_0 w_j=g_j$ weakly with Neumann data. Then
$w_j\geq0$ and $\LL_tw_j\geq\rho_j$ for all loading times covered by
the prefix. A computed Neumann spline $z_j$ supplies the fully computable
enclosure
\begin{equation}\label{snum:aux-correction}
 |w_j-z_j|\leq\frac{\|g_j-\LL_0z_j\|_\infty}{a}.
\end{equation}
To verify \eqref{snum:aux-correction}, set
$R_j=g_j-\LL_0z_j$, $h=w_j-z_j$, and
$\delta=\|R_j\|_\infty/a$. Then $\LL_0h=R_j$. For
$\eta_\pm=(\pm h-\delta)_+$, the weak equations give
\begin{align}
 0\leq\|\eta_\pm\|_0^2
 &=A_0(\pm h-\delta,\eta_\pm)
 =(\pm R_j-a\delta,\eta_\pm)\notag\\
 &\leq(\|R_j\|_\infty-a\delta)\int_\Omega\eta_\pm\dd x=0.
 \label{eq:auxiliary-correction-proof}
\end{align}
Thus $-\delta\leq h\leq\delta$. The corrected spline bounds enclose
the exact nonnegative supersolution $w_j$.

For a stagewise construction, let $E_0$ be a nonnegative spatial constant
dominating the preceding channel supersolution $e_{\mathrm{prev}}$ at
the stage entrance. Suppose $a+f\geq a_*>0$ throughout the stage and
$g\geq\rho$.
For a Neumann spline candidate $z$ set
\begin{equation}\label{snum:stage-correction}
 R=g+b\Delta z-a_*z,\qquad
 \delta=\max\{\|(R)_+\|_\infty/a_*,\ E_0-\inf_\Omega z,\ 0\},
 \quad v=z+\delta.
\end{equation}
The correction satisfies
\begin{align}
 v&\geq\inf_\Omega z+\delta\geq E_0\geq e_{\mathrm{prev}}\geq0,\notag\\
 \LL_tv
 &=(-b\Delta+a_*)v+(a+f-a_*)v\notag\\
 &=g-R+a_*\delta+(a+f-a_*)v
 \geq g-R+\|(R)_+\|_\infty\geq g\geq\rho.
 \label{eq:stage-inheritance-proof}
\end{align}
This proves stage inheritance and the supersolution inequality.
Section~\ref{sec:sen} takes $E_0$ from the first-stage constant prefix
radius.
\subsection{Verified evaluation from coefficients to the acceptance test}
\label{sec:verified-arithmetic}\label{app:2d-arithmetic}\label{app:arithmetic}

We evaluate the residual, source, and norm bounds using polynomial
enclosures and outward arithmetic.

\emph{Polynomial inclusion and exact integration.}
Map each spatial cell affinely to the unit interval or square.
For $B_i^m(s)=\binom mi s^i(1-s)^{m-i}$, positivity and partition of
unity give the convex-hull bound \cite{Farouki2012},
\begin{equation}\label{eq:bernstein-convex-chain}
 \min_i c_i\leq\sum_{i=0}^m c_iB_i^m(s)\leq\max_i c_i,
 \qquad 0\leq s\leq1.
\end{equation}
The tensor weights $B_i^m(s)B_j^n(t)$ have the same properties.
The product and integration identities are
\begin{align}
 B_i^mB_k^m
 &=\frac{\binom mi\binom mk}{\binom{2m}{i+k}}B_{i+k}^{2m},\notag\\
 \int_0^1B_\ell^{2m}(s)\dd s
 &=\binom{2m}{\ell}\frac{\ell!(2m-\ell)!}{(2m+1)!}
   =\frac1{2m+1},\notag\\
 G^m_{ik}:=\int_0^1 B_i^mB_k^m\dd s
 &=\frac{\binom mi\binom mk}{(2m+1)\binom{2m}{i+k}}.
 \label{eq:bernstein-mass-chain}
\end{align}
The middle line is the beta integral for nonnegative integer exponents.
Consequently, for $q|_K=\sum_{i,j}q_{ij}B_i^m(s)B_j^n(t)$,
\begin{align}
 \int_K q^2\dd x\dd y
 &=|K|\int_0^1\!\int_0^1
    \sum_{i,j,k,l}q_{ij}q_{kl}B_i^m(s)B_k^m(s)B_j^n(t)B_l^n(t)
       \dd s\dd t\notag\\
 &=|K|\sum_{i,j,k,l}q_{ij}q_{kl}G^m_{ik}G^n_{jl}.
 \label{eq:tensor-mass}
\end{align}
In one dimension, collecting the terms with $i+j=\ell$ in
Eq.~\eqref{eq:bernstein-mass-chain} gives the coefficients of
$q^2$, where $q=\sum_i c_iB_i^m$,
\[
 \beta_\ell=\sum_{i+j=\ell}
 \frac{\binom{m}{i}\binom{m}{j}}{\binom{2m}{\ell}}c_ic_j .
\]
The squared norm is
\begin{equation}\label{eq:poly-norm}
 \norm{q}_{L^2(\Omega)}^2
 =\sum_K\frac{|K|}{2m+1}\sum_{\ell=0}^{2m}\beta_{K,\ell}.
\end{equation}
Products, derivatives, restrictions, and degree elevations use outward
interval operations. Coefficient products retain their signs when
evaluating \eqref{eq:tensor-mass} and \eqref{eq:poly-norm}. A squared-norm
enclosure $[s_-,s_+]$ gives norm bounds by outward evaluation of
$[\sqrt{\max(0,s_-)},\sqrt{s_+}]$, including the lower bound $q_-$.
The midpoint of the enclosed raw polynomial coefficients is the
delivered field. Its coefficient radius is included in
$B_{\mathrm{fp},K}$.

The same representation controls variation between loading endpoints.
For a clock with Bernstein coefficients $c_0,\ldots,c_m$, differentiation
gives
$C'(s)=m\sum_{i=0}^{m-1}(c_{i+1}-c_i)B_i^{m-1}(s)$.
Apply the triangle inequality and integrate the nonnegative basis
functions, each with integral $1/m$. The result is
\[
 \int_0^1|C'(s)|\dd s\leq
 \sum_{i=0}^{m-1}|c_{i+1}-c_i|.
\]
Applied after restriction to $I_n$, this gives $V_{K,n}$ in
\eqref{eq:source-channel} and retains interior pulse variation even when
$c_m=c_0$.

\emph{Rational data and elementary operations.}
Dyadic knots and loading endpoints give rational knot-insertion,
power-to-Bernstein, and mass matrices. Each rational entry is enclosed
by comparing it with the exact rational value of its proposed binary64
approximation and widening the appropriate endpoint.

Basic interval operations move rounded endpoints to adjacent
representable values outward. The backend assumes IEEE-754 binary64
with gradual underflow and correctly rounded operations and square root.
Inputs and results must be finite, and every divisor interval must be
separated from zero.

\emph{Signed tensor contractions.}
Let $0\leq m_i^-\leq m_i\leq m_i^+$ and
$p_i^-\leq p_i\leq p_i^+$ for $i=1,\ldots,k$. Write
$z^+=\max(z,0)$ and $z^-=\min(z,0)$ in this paragraph.
The exact lower and upper endpoint expressions are
\begin{align}
 L&=\sum_i\{m_i^-(p_i^-)^++m_i^+(p_i^-)^-\},\notag\\
 U&=\sum_i\{m_i^+(p_i^+)^++m_i^-(p_i^+)^-\},\qquad
 L\leq\sum_i m_ip_i\leq U.\label{eq:signed-dot-endpoints}
\end{align}
The sign of the input determines the matrix endpoint in each product.
Let $u_{\rm fp}=2^{-53}$ be unit roundoff, $\epsilon=2u_{\rm fp}$, and let $\eta=2^{-1074}$
be the smallest positive subnormal. Each finite elementary result obeys
\begin{equation}\label{eq:rounding-model}
 \operatorname{fl}(z)=z(1+\delta)+\zeta,
 \qquad |\delta|\leq u_{\rm fp},\quad |\zeta|\leq\eta/2.
\end{equation}
A signed endpoint uses at most $N=4k+1$ multiply/add operations.
After a product is formed, the remaining operations in each reduction
are additions. Each local absolute underflow error therefore propagates
through the subsequent addition-rounding factors $1+\delta$. Expanding the computation graph gives, for either endpoint
$S\in\{L,U\}$ and its computed value $\widehat S$,
\begin{align}
 |\widehat S-S|
 &\leq\big[(1+u_{\rm fp})^N-1\big]H
       +\frac{N\eta}{2}(1+u_{\rm fp})^N\notag\\
 &\leq\frac{Nu_{\rm fp}}{1-Nu_{\rm fp}}H+\frac{N\eta}{2(1-Nu_{\rm fp})},
 \qquad H=\sum_i m_i^+\max(|p_i^-|,|p_i^+|).
 \label{eq:dot-error-expansion}
\end{align}
There is at most one nonzero signed part per input endpoint, so $H$
bounds the sum of absolute exact products. The second inequality uses
$(1+u_{\rm fp})^N\leq(1-Nu_{\rm fp})^{-1}$ for $Nu_{\rm fp}<1$. It follows, for example, by
comparing the binomial series termwise with the geometric series.

The implementation uses $s=8k+8$ under the condition $s\epsilon<1/2$.
Then $N<s/2$ and $Nu_{\rm fp}<s\epsilon/4<1/8$, giving
\begin{equation}\label{eq:dot-conservative-chain}
 \frac{Nu_{\rm fp}}{1-Nu_{\rm fp}}\leq\gamma_s:=\frac{s\epsilon}{1-s\epsilon}<1,
 \qquad \frac{N}{2(1-Nu_{\rm fp})}<\frac{2s}{7}<s.
\end{equation}
Equations~\eqref{eq:dot-error-expansion}--\eqref{eq:dot-conservative-chain}
therefore yield $|\widehat S-S|\leq\gamma_sH+s\eta$, including
gradual underflow. The nonnegative contraction computing $\widehat H$
has a smaller operation count and satisfies the same bound. Rearranging
its lower error inequality gives
\begin{align}
 \widehat H&\geq(1-\gamma_s)H-s\eta,\notag\\
 H&\leq H^+:=\frac{\widehat H+s\eta}{1-\gamma_s},\qquad
 |\widehat S-S|\leq M:=\gamma_sH^++s\eta.\label{eq:dot-margin-chain}
\end{align}
All quantities on the right are evaluated outward, using upper bounds
for $\gamma_s,s\eta$ and a positive lower bound for $1-\gamma_s$.
Subtracting $M$ from $\widehat L$ downward and adding $M$ to
$\widehat U$ upward gives an interval containing the exact contraction.
The bound is valid for any reduction ordering satisfying the stated
operation count.

\emph{Inclusion through the complete certificate.}
Nonnegative extraction, elevation and tensor multiplication maps preserve
coefficient inclusion by Eq.~\eqref{eq:signed-dot-endpoints} and its
rounding enclosure. Differentiation consists of interval differences
and exact degree factors, with the physical cell-width scaling retained.
The full loading polynomial, or the exact affine-clock
identity in Eq.~\eqref{eq:2d-loading-residual}, encloses interval interiors.
The endpoint defects in Eq.~\eqref{eq:defect} are accumulated as signed
interval coefficients. The mass identity, pairwise interval summation
over cells, and outward square root enclose their norms and the raw norm.
Finally, the numerator $C$ is rounded upward and $q_--C$ downward;
acceptance requires a positive denominator endpoint and an upward-rounded
ratio below the prescribed tolerance.
The one-dimensional Green realization uses the accumulation bound
\cite{Higham2002}
\begin{equation}
 \gamma_{2n}=\frac{2n\epsilon_{\rm mach}}{1-2n\epsilon_{\rm mach}},
\end{equation}

together with an absolute term for subnormal operations.
Here $n$ is the dot-product length. Using machine epsilon in
$\gamma_{2n}$ gives a conservative round-to-nearest accumulation bound.
The absolute term is based on at most $2n$ operations and the
smallest positive binary64 subnormal number $\eta$.

The preserved one-dimensional runs used the relative part of this bound.
Their separate audit propagates
$2n\eta/(1-\gamma_{2n})$ through the three channel barriers and
the source transfer. For $0\leq u\leq1$ the function
$g(E)=\min\{E,(1-u)E+E^2/2\}$ is at most 2-Lipschitz. The quadratic branch is active for $0\leq E\leq2u$,
where its derivative is $1-u+E\leq1+u\leq2$, and the linear branch
has derivative 1. The two branches agree at $E=2u$.
The exact loading weight is bounded by
$f(T)+\TV(f)\leq3\sum_jA_j$ for the specified clocks.
These facts give an explicit uniform absolute source margin.
The maximum margin is $2.94\times10^{-318}$ in the primary study
and $3.33\times10^{-318}$ in the additional cases.
The tabulated certificates include these margins.

\section{Adaptive Allocation of Numerical Accuracy}\label{sec:algorithm}
\begin{figure}[t]
\centering\includegraphics[width=.95\linewidth]{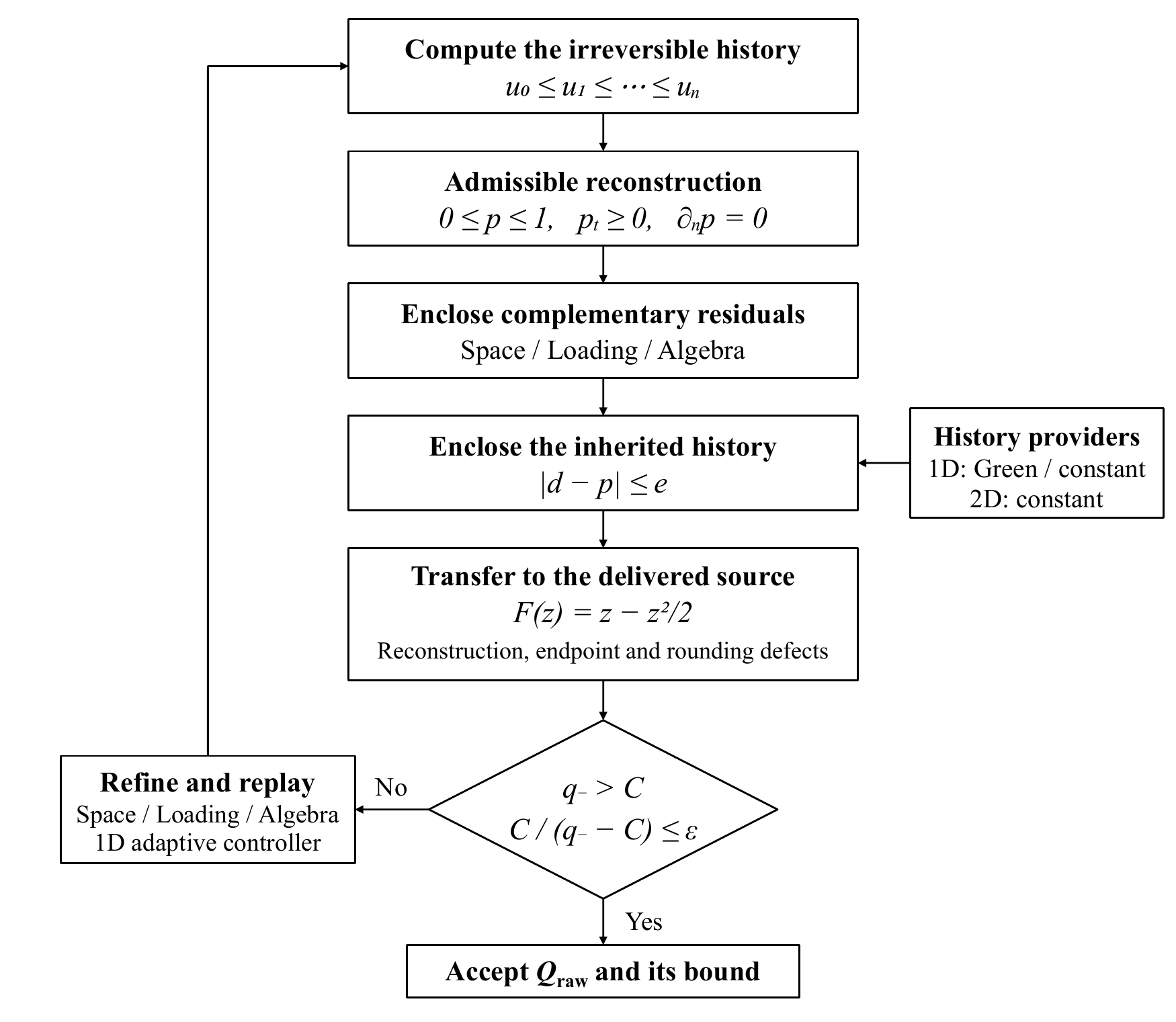}
\caption{Causal certification and complete-history adaptive replay.}
\label{fig:workflow}
\end{figure}

The absolute budget follows from \eqref{eq:relative-certificate} and
applies to both realizations. The adaptive controller uses the
one-dimensional Green provider. Set
\begin{equation}\label{eq:budget}
 \beta=\frac{\varepsilon}{1+\varepsilon}\,q_-.
\end{equation}
For $\varepsilon>0$ and $q_->C$, the acceptance test is
\begin{equation}\label{eq:budget-equivalence}
 \frac{C}{q_--C}\leq\varepsilon
 \quad\Longleftrightarrow\quad C\leq\beta.
\end{equation}
The balanced CG policy starts with a normalized projected-residual tolerance
$10^{-3}$. If $\norm{B_a}_{L^2}>0.1\beta$, it divides that tolerance by ten
and repeats the history. Otherwise, it chooses between the spatial norm
$\norm{B_x}_{L^2}$ and the loading norm
$\norm{B_\tau}_{L^2}+\norm{B_q}_{L^2}$.
The direct policy uses tight banded solves and the same spatial/loading
choice. A uniform baseline refines all cells or all loading intervals in
the selected channel, keeping the stopping certificate unchanged.

To localize refinement, let $\overline G_{KJ}$ majorize the response on cell
$K$ to a unit constant forcing on cell $J$, including integration over $J$.
With $W_K=f_K^T+\sum_nV_{K,n}$, define
\begin{equation}\label{eq:sensitivity}
 s_J=\left(\sum_K |K|(W_K\overline G_{KJ})^2\right)^{1/2}.
\end{equation}
A spatial score combines the residual and recovery terms
\[s_J\max_n\rho_{x,J,n}\quad\text{and}\quad
 |J|^{1/2}W_J\max_n r^{\mathrm{rec}}_{J,n}.\]
A loading score combines the computed $L^2$ norm of the local polynomial
defect $D_n$ with the residual score
\[\left(\sum_J(s_J\rho_{\tau,J,n})^2\right)^{1/2}.\]
These scores emphasize residuals that influence the accumulated source
through the positive elliptic response and the variation of the loading.
The largest scores are marked until at least 60\% of their sum is covered.
Marked cells or intervals are bisected. Neighboring spatial cells are
graded to a size ratio of at most two.

After each update, the controller replays the history from the exact
initial state on the new discretization and recomputes $q_-$ and
$\beta$. It continues from $q_-\leq C$ until the positive-denominator
and accuracy tests both hold. All work counts include every replay and
restricted solve. A residual-evaluation floor supplements the requested
algebraic tolerance, and the actual residual enters the certificate.

The algebraic budget follows inexact-solve balancing
\cite{Jiranek2010,Haberl2021}, and the 60\% selection follows bulk marking
\cite{Dorfler1996}. Figure~\ref{fig:workflow} and
Algorithm~\ref{alg:control} summarize the controller.

\begin{algorithm}[!htbp]
\caption{Certified allocation for the raw cumulative source}\label{alg:control}
\begin{algorithmic}[1]
\Require Prescribed $f$ on every loading interval, initial state, target
$\varepsilon$, initial spatial and loading partitions
\State Set the algebraic tolerance and initialize cumulative work counters.
\Loop
 \State Solve the complete irreversible history and add all work to the counters.
 \State Recover the nondecreasing spline history and compute complementary
 residual envelopes.
 \State Form prefix barriers, channel source bounds, the raw defect and
 outward rounding bounds.
 \State Compute $Q_{\raw}$, $C$, $q_-$ and $\beta=\varepsilon q_-/(1+\varepsilon)$.
 \If{$q_->C$ and $C\leq\beta$}
  \State \Return the unchanged $Q_{\raw}$, the cellwise bounds and the history.
 \ElsIf{CG is used and $\norm{B_a}_{L^2}>0.1\beta$}
  \State Reduce the algebraic tolerance by a factor of ten.
 \Else
  \State Mark the dominant spatial or loading channel and bisect.
  \State Apply neighboring-cell grading after spatial refinement.
 \EndIf
\EndLoop
\end{algorithmic}
\end{algorithm}

Reliability follows from the complete acceptance test. Estimator
stability, reduction, discrete reliability, and quasi-orthogonality
underpin convergence and complexity analysis \cite{Carstensen2014};
recent adaptive FEM results incorporate inexact solvers into optimal-cost
algorithms \cite{Bringmann2025}. These results identify the analytical
properties to investigate for the source-weighted marking scores.
The two-dimensional studies use nested uniform meshes.
\FloatBarrier

\section{Numerical Assessment}\label{sec:numerical}

The study tests the delivered-output bound along adaptive histories,
isolates the accuracy channels, and measures computational cost.
Acceptance uses the production certificate. Separately enclosed
references assess error and effectivity after computation, including
rejected coarse and intermediate histories.

\subsection{Test problems and source accuracy}
To evaluate the path-area mechanism in Fig.~\ref{fig:history}(a) exactly, take
\begin{equation}\label{eq:counterexample}
 K=\begin{pmatrix}5/4&-1/4\\-1/4&5/4\end{pmatrix},
 \qquad d_0=(1/5,1/5)^\top,\qquad d_T=(3/5,3/5)^\top.
\end{equation}
Increasing the first component and then the second gives
$q^{AB}=(9/50,7/50)^\top$. Reversing the order gives
$q^{BA}=(7/50,9/50)^\top$. Both paths are irreversible and have the same
terminal state and total source $8/25$. Nonnegative drives $f_i=(Kd)_i/(1-d_i)$ realize both paths as unique
stationary minima of the strictly convex energy.

The six problem classes are uniform monotone loading (U), a uniform pulse
(U$_p$), sequential localized loading (S), two local pulses (P), narrow
localization (N), and overlapping loading (O). All use $T=1$, $a=1$, and $\Omega=(0,1)$ with homogeneous
Neumann conditions and zero initial damage. The uniform histories have exact solutions; the localized cases introduce
spatial coupling and different loading variations.

Define the monotone clock $S(s)=s^3(10-15s+6s^2)$ on $[0,1]$, extended
by zero and one to the left and right, respectively. Define the pulse
$P(s)=64s^3(1-s)^3$ on $[0,1]$, extended by zero. The spatial shape
$\phi_{[\alpha,\beta]}(x)$ is the pulse polynomial with argument
$(x-\alpha)/(\beta-\alpha)$ and vanishes outside its support; the constant
shape is $\phi_1=1$. The prescribed drive is
\begin{equation}\label{eq:load-family}
 f(x,t)=\sum_j A_j\phi_j(x)
 C_j\!\left(\frac{t-\tau_j^-}{\tau_j^+-\tau_j^-}\right),
 \qquad C_j\in\{S,P\}.
\end{equation}
Table~\ref{tab:cases} specifies the amplitudes, supports, clock intervals,
and interaction lengths. Terms with the same case label are added,
including the monotone background in each localized drive.
Figures~\ref{fig:uniform} and \ref{fig:loads} show the resulting loading
families. These parameters remain fixed throughout refinement.

\begin{table}
\centering
\small
\caption{Prescribed driving fields. Each row defines one term in Eq.~\eqref{eq:load-family}. Spatial support 1 denotes the constant shape. $S$ is the monotone clock and $P$ is the pulse. All cases use $a=1$, $b=k^{-2}$, $T=1$, and zero initial damage.}
\label{tab:cases}
\begin{tabular}{lrrrrc}
\toprule
Case & $k$ & $A_j$ & Spatial support & Clock interval & Clock\\
\midrule
Uniform loading & 32 & 4 & 1 & [0,1] & S\\
Uniform pulse & 32 & 4 & 1 & [0,1] & P\\
Sequential loading & 32 & 0.125 & 1 & [0,1] & S\\
 &  & 8 & [0.25,0.5] & [0,0.5] & S\\
 &  & 8 & [0.5,0.75] & [0.5,1] & S\\
Two local pulses & 32 & 0.125 & 1 & [0,1] & S\\
 &  & 8 & [0.25,0.5] & [0,0.75] & P\\
 &  & 8 & [0.5,0.75] & [0.25,1] & P\\
Narrow localization & 64 & 0.125 & 1 & [0,1] & S\\
 &  & 16 & [0.3125,0.4375] & [0.125,0.375] & S\\
 &  & 12 & [0.5625,0.6875] & [0.625,0.875] & S\\
Overlapping loading & 16 & 0.125 & 1 & [0,1] & S\\
 &  & 12 & [0.25,0.625] & [0,0.75] & S\\
 &  & 8 & [0.375,0.75] & [0.25,1] & S\\
\bottomrule
\end{tabular}
\par\smallskip
\end{table}

\begin{figure}
\centering\includegraphics[width=\linewidth]{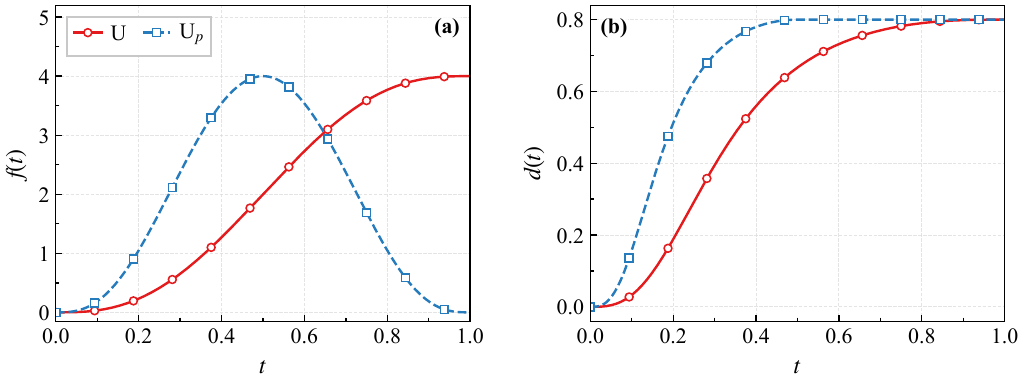}
\caption{Uniform loading histories: (a) prescribed drive; (b) irreversible damage.
U denotes monotone loading and U$_p$ a pulse with the same maximum drive.}\label{fig:uniform}
\end{figure}
\begin{figure}
\centering\includegraphics[width=\linewidth]{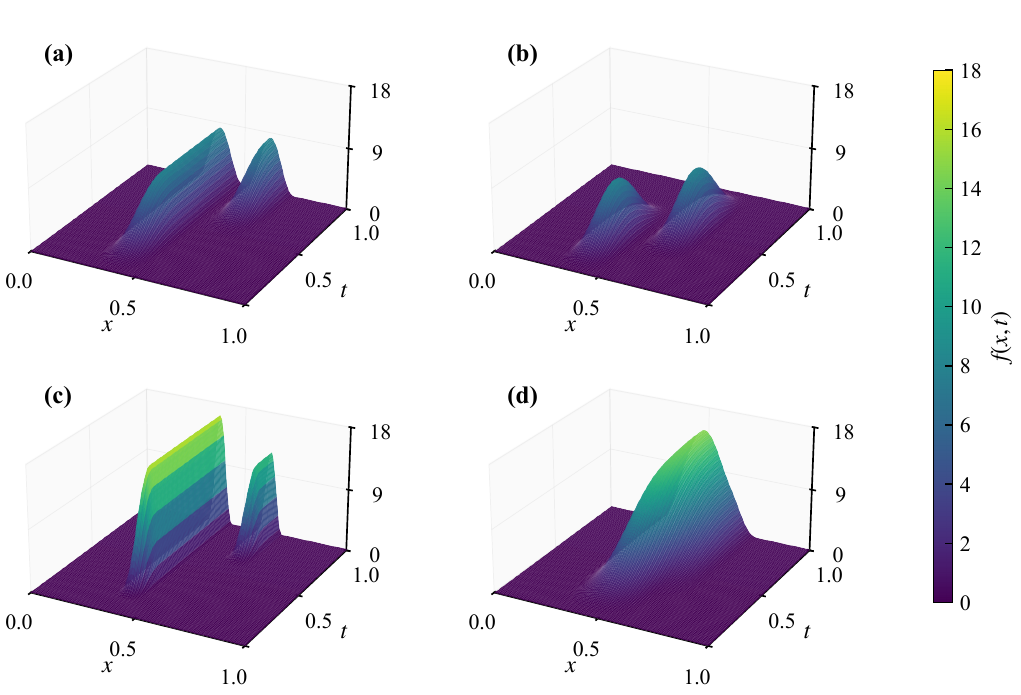}
\caption{Localized driving fields $f(x,t)$ from Eq.~\eqref{eq:load-family}:
(a) sequential loading S; (b) two pulses P; (c) narrow localization N;
(d) overlapping loading O. All panels share the physical scale.}\label{fig:loads}
\end{figure}

The primary study combines U, U$_p$, S and P with the targets
$\varepsilon=0.05$ and $0.01$ and three policies: source marking with
balanced CG, uniform refinement of the selected channel with tight CG, and
source marking with tight direct solves. These combinations give 24 runs.
The N and O problems use the same source controller and direct solver at
both targets, giving four additional runs with parameters fixed beforehand.

Initial partitions contain 16 spatial cells and eight loading intervals,
augmented by prescribed breakpoints and pulse turning points. A separate pulse control starts with its two endpoints alone.
The implementation allows at most 60 complete-history rounds, 2048 spatial
cells and 4096 loading intervals. All 28 runs attain their prescribed target
within these limits.

For a uniform drive, the exact state is
$d(t)=m(t)/(1+m(t))$, where $m(t)=\max_{s\leq t}f(s)$. Both uniform
histories attain $m(T)=4$. During damage growth, $f(1-d)=d$, giving
\begin{equation}\label{eq:exact-uniform}
 Q(x,T)=\tfrac12 d(T)^2=\frac{8}{25}.
\end{equation}
The monotone benchmark U also gives an exact identity for the endpoint
accumulation. Here $a=1$, $d_0=0$, and $d_n=d(t_n)$ denotes the exact
uniform endpoint state. During each positive damage increment,
$f_n(1-d_n)=d_n$; a zero increment contributes zero to either expression.
Writing $\Delta d_n=d_n-d_{n-1}$ and using
\[
 2d_n\Delta d_n=d_n^2-d_{n-1}^2+(\Delta d_n)^2,
\]
the exact-arithmetic rule \eqref{eq:raw} therefore gives
\begin{align}\label{snum:uniform-endpoint}
 \widetilde Q_{\raw}
 &=\sum_n f_n(1-d_n)\Delta d_n=\sum_n d_n\Delta d_n\notag\\
 &=\frac12d(T)^2+\frac12\sum_n(\Delta d_n)^2.
\end{align}
The excess over \eqref{eq:exact-uniform} is
$\tfrac12\sum_n(\Delta d_n)^2$. The identity also holds for U$_p$
when its maximum is a loading node, since every positive increment ends
at a running maximum and unloading contributes zero.

The localized references carry independently computed error bounds.
Monotone cases use a high-order boundary-value solver followed by polynomial
source integration and a residual enclosure. The absolute source radii are
$8.4970\times10^{-6}$ for S and its reversed loading order,
$8.2566\times10^{-6}$ for N, and $1.1395\times10^{-7}$ for O.
For P, a tightly solved P1 obstacle reference and an independent cubic
B-spline collocation obstacle reference have radii
$4.1782\times10^{-4}$ and $4.5435\times10^{-4}$, respectively.
Exact polynomial integration of the reconstructed histories gives centers
that differ by at most $1.637\times10^{-6}$ in $L^2$. Selecting the
smaller reference radius cellwise on a common refinement gives the
combined radius $3.4872\times10^{-4}$.

If $\underline\Delta$ and $\overline\Delta$ enclose
$\|Q_{\raw}-Q_{\refe}\|_{L^2}$, independent verification gives
\begin{equation}\label{eq:verification}
 E_-:=\max(0,\underline\Delta-C_{\refe})
 \leq\|Q_{\raw}-Q\|_{L^2}
 \leq\overline\Delta+C_{\refe}=:E_+.
\end{equation}
An enclosed lower bound for the exact source norm gives the verified
relative upper bound. Effectivity is enclosed by $[C/E_+,C/E_-]$ when
$E_->0$, with an unbounded upper endpoint when $E_-=0$.

The 25 fixed-parameter calculations comprise seven loading partitions
for each uniform case, five spatial meshes for S, and six CG tolerances
for S. Companion parameters are fixed within each scan, and every complete
bound covers the independently enclosed error. The loading scans use
$M=16$, direct solves, and successive doubling of $N$ from 8 to 512
(Fig.~\ref{fig:loading-scan}). The verified error upper bound decreases from
$6.9558\times10^{-2}$ to $1.1252\times10^{-3}$ in U and from
$1.1278\times10^{-1}$ to $2.0592\times10^{-3}$ in U$_p$.
The corresponding final certificates are $1.1444\times10^{-3}$ and
$2.1209\times10^{-3}$. At finer partitions the complete bound approaches the endpoint-rule
contribution, consistent with \eqref{snum:uniform-endpoint}.

\begin{figure}
\centering\includegraphics[width=\linewidth]{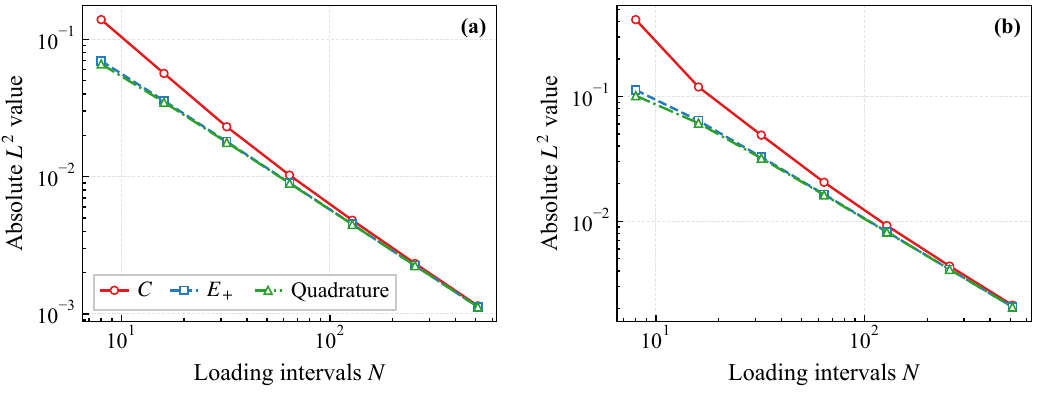}
\caption{Loading refinement with $M=16$ and direct solves: (a) U; (b) U$_p$.
$C$ and $E_+$ denote the complete bound and verified error upper endpoint;
Quadrature denotes the endpoint-rule contribution.}\label{fig:loading-scan}
\end{figure}

For spatial refinement in S at $N=512$, increasing $M$ from 32 to 512
reduces the complete bound from $1.6981\times10^{-1}$ to
$1.8330\times10^{-3}$ (Fig.~\ref{fig:space-solver-scan}a).
The spatial contribution continues to fall as the verified error approaches
the loading floor near $1.52\times10^{-3}$. The algebraic scan retains
$M=128$, $N=256$ and decreases the requested CG tolerance from $10^{-1}$
to $10^{-9}$ (Fig.~\ref{fig:space-solver-scan}b). The complete bound
decreases from $0.23086$ to $0.0090435$, with a marked reduction in the
algebraic contribution. Below approximately $10^{-4}$ the verified source
error remains near $0.00318$. Small nonmonotone error changes reflect signed cancellation.
The spatial and algebraic scans identify the loading floor beyond which
further refinement of the selected channel gives little benefit.

\begin{figure}
\centering\includegraphics[width=\linewidth]{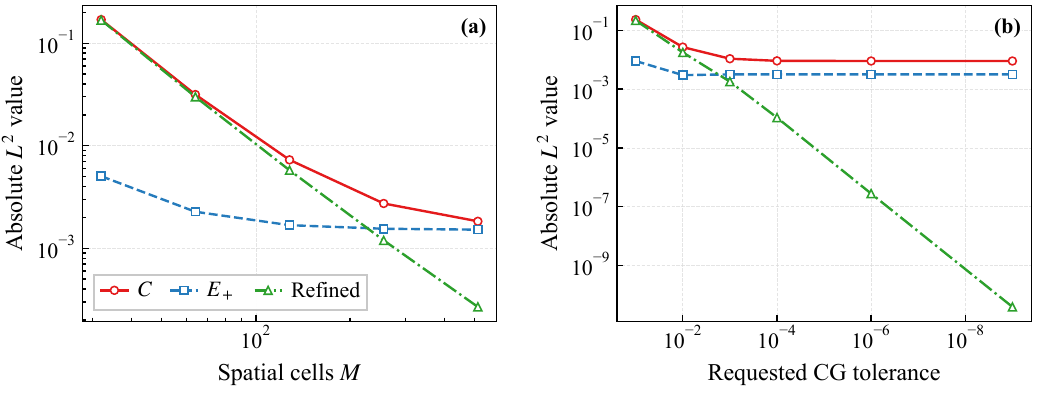}
\caption{Sequential loading S: (a) spatial refinement at $N=512$;
(b) CG tolerance refinement at $M=128$, $N=256$.
Refined denotes $\|B_x\|_2$ in (a) and $\|B_a\|_2$ in (b).}\label{fig:space-solver-scan}
\end{figure}

Table~\ref{tab:accuracy} reports the initial source-error intervals and
the results at both requested targets for all six direct-solver cases.
The initial errors range from approximately 21.74\% in U to 158.78\% in
N. Each is measured relative to the norm of the same continuous source
used to verify its final result. At the 5\% stopping target, every final
verified upper bound lies below 5\%; tightening the target to 1\% gives
certificates between 0.773\% and 0.979\% and verified upper bounds between
0.459\% and 0.913\%.
\begin{table}
\centering
\caption{Original source-error reduction at the 5\% and 1\% targets with source marking and direct solves. All entries are percentages.}
\label{tab:accuracy}
\small
\begin{tabular}{llrrrrr}
\toprule
Case & Initial error interval & $C_{5\%}^{\rm rel}$ & $E_{5\%,+}^{\rm rel}$ & $C_{1\%}^{\rm rel}$ & $E_{1\%,+}^{\rm rel}$ & Reduction\\
\midrule
U & $[21.736,21.737]$ & 4.346 & 3.811 & 0.954 & 0.913 & 82.469\\
U$_p$ & $[35.244,35.245]$ & 3.462 & 3.135 & 0.773 & 0.750 & 91.107\\
S & $[44.819,44.831]$ & 4.857 & 2.111 & 0.875 & 0.460 & 95.290\\
P & $[54.419,54.911]$ & 4.875 & 1.969 & 0.979 & 0.606 & 96.382\\
N & $[158.767,158.792]$ & 4.447 & 2.125 & 0.829 & 0.481 & 98.661\\
O & $[31.202,31.203]$ & 4.396 & 1.992 & 0.823 & 0.488 & 93.618\\
\bottomrule
\end{tabular}
\par\smallskip\begin{minipage}{.97\linewidth}\footnotesize $C^{\rm rel}=100C/(q_--C)$ is the production certificate and $E_+^{\rm rel}$ the verified error upper bound. Reduction uses the initial error lower endpoint and the final upper endpoint at 5\%. All bounds are rounded outward.\end{minipage}
\end{table}

The narrow-localization case gives the most pronounced transition.
Its initial 16-cell, eight-interval calculation has a relative error
enclosed in $[158.767,158.792]\%$. After 21 refinement decisions, the
controller accepts 187 cells and 170 loading intervals with a certificate
below 4.447\%. Independent verification bounds the actual relative error
by 2.125\%, a reduction exceeding 98.66\% from the initial error.
Refinement to the 1\% target subsequently gives a certificate below
0.829\% and a verified error below 0.481\%. The large initial error reflects underresolution of the concentrated
source. The load, constitutive parameters, and endpoint rule remain fixed.

Figure~\ref{fig:convergence} tracks the verified error through all
refinement rounds, including initial states with $q_-\leq C$.
The accepted certificates are marked at the final round.

\begin{figure}
\centering\includegraphics[width=\linewidth]{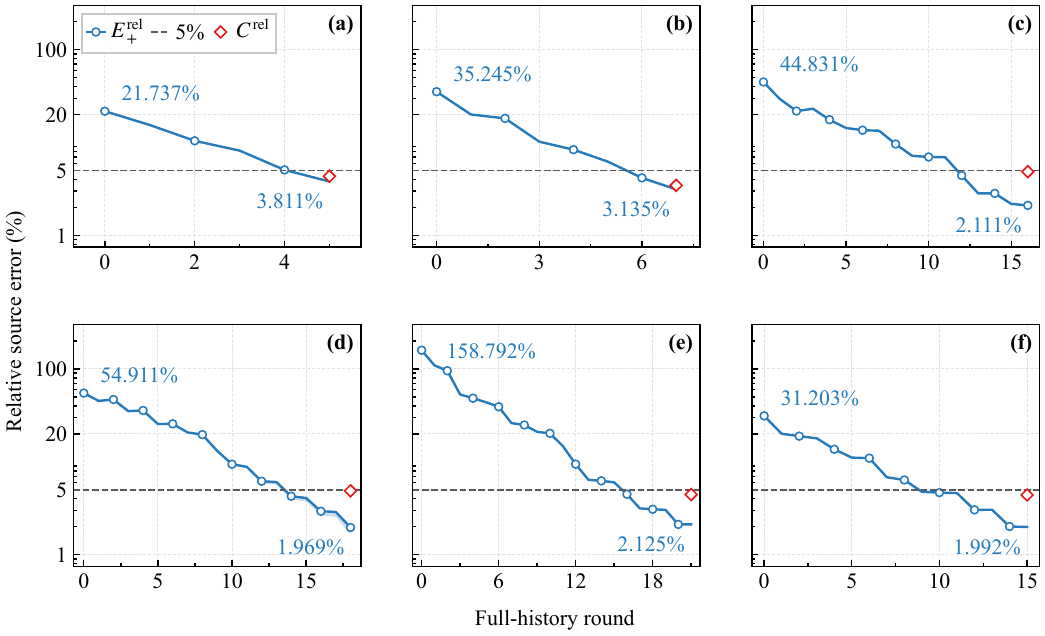}
\caption{Verified source-error reduction at the 5\% target: (a) U; (b) U$_p$;
(c) S; (d) P; (e) N; (f) O. Bands retain reference uncertainty;
diamonds mark accepted certificates. In (e), the error falls from
approximately 158.78\% to below 2.125\%.}
\label{fig:convergence}
\end{figure}

\subsection{Localized history, source distribution and spatial enclosure}
Figure~\ref{fig:damage} shows the accepted histories. Sequential loading
activates neighboring regions in turn; pulses retain damage during
unloading; narrow and overlapping loads change the spatial concentration
of growth. These differences affect the complementary residual and
loading-variation weights in \eqref{eq:source-channel}.

\begin{figure}
\centering\includegraphics[width=\linewidth]{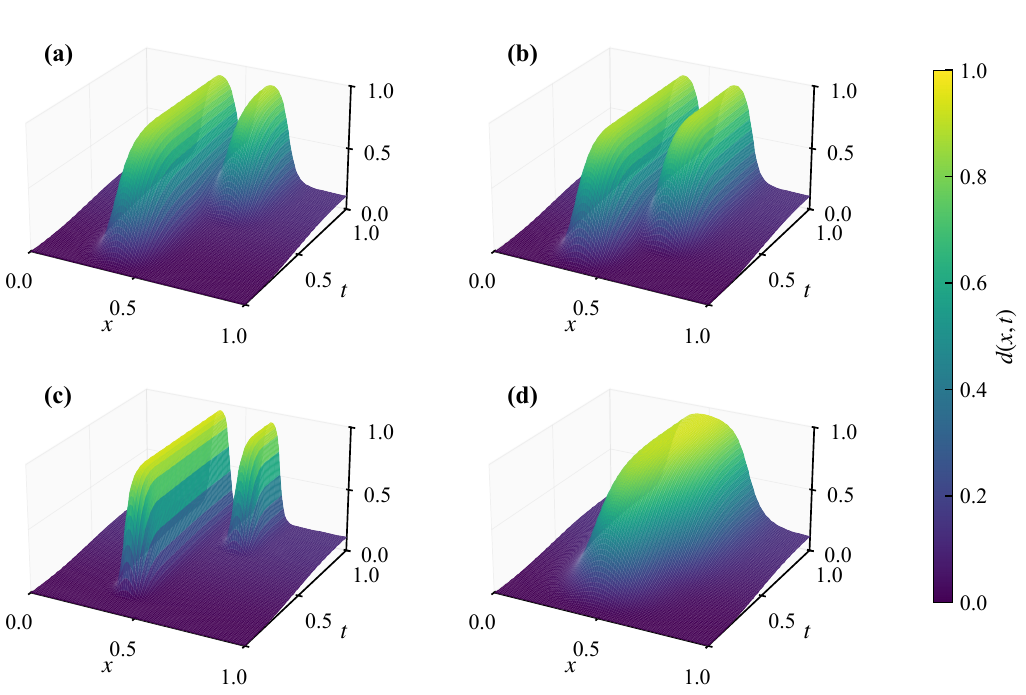}
\caption{Accepted damage histories at the 1\% direct-solver target:
(a) S; (b) P; (c) N; (d) O. All panels share the range $0\leq d\leq1$.}
\label{fig:damage}
\end{figure}

The accumulated fields agree closely with their enclosed references
(Fig.~\ref{fig:sources}). Insets resolve the remaining differences near
source maxima. The two-pulse calculation uses 474 cells and 796 loading
intervals, with certificate below 0.979\% and verified relative error
below 0.606\%.

\begin{figure}
\centering\includegraphics[width=\linewidth]{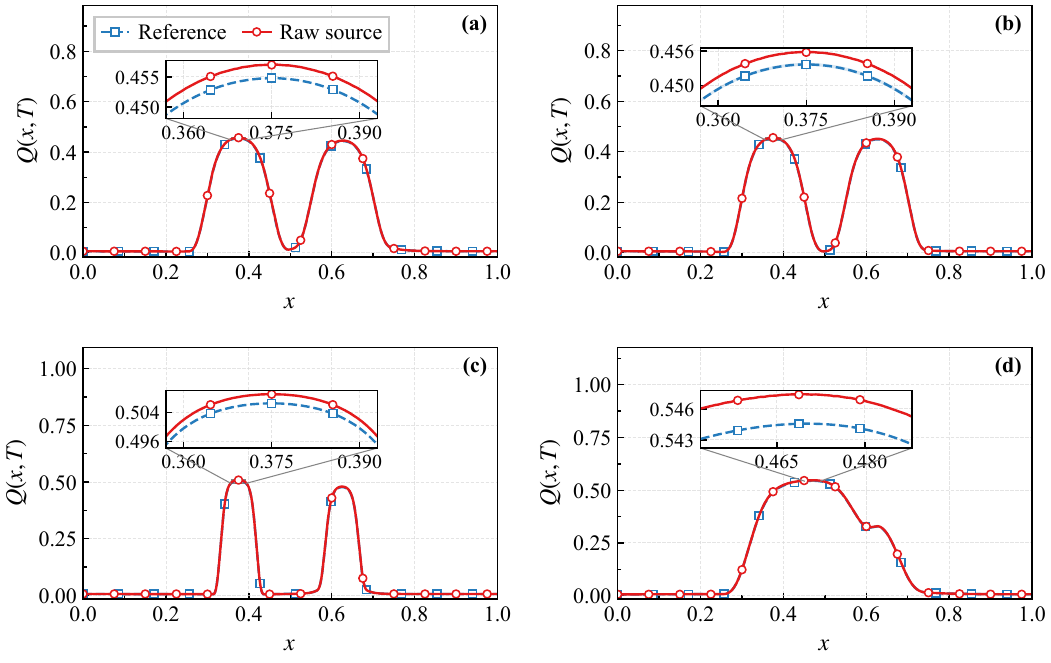}
\caption{Raw source and independently enclosed reference at the 1\% direct-solver
target: (a) S; (b) P; (c) N; (d) O. Insets enlarge the reference-peak region;
bands show reference uncertainty.}
\label{fig:sources}
\end{figure}

The pointwise production bounds contain the source differences and
reference uncertainty in Fig.~\ref{fig:enclosures}. Their wider regions
follow localized growth and source sensitivity. In the pulse case,
two separated error regions persist after unloading, while the
complementary residual accounts for the reaction on frozen cells.

\begin{figure}
\centering\includegraphics[width=\linewidth]{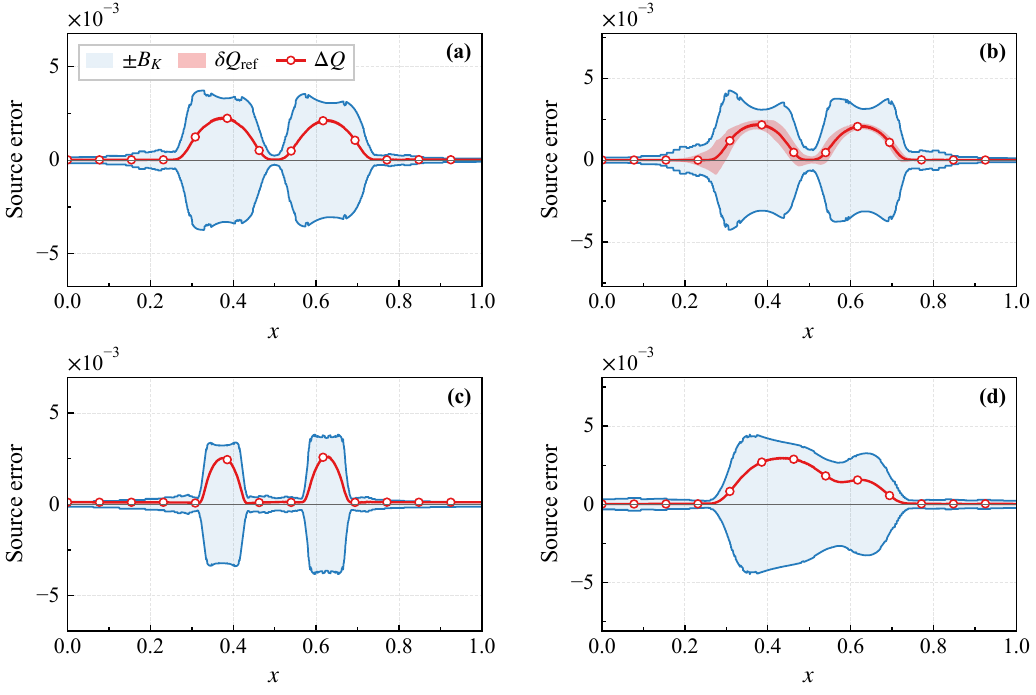}
\caption{Pointwise source-error enclosures: (a) S; (b) P; (c) N; (d) O.
$\Delta Q=Q_{\raw}-Q_{\refe}$ is shown with reference radius
$\delta Q_{\refe}$ and production bounds $\pm B_K$.}
\label{fig:enclosures}
\end{figure}

The saved meshes in Fig.~\ref{fig:resolution} show repeated refinement
around localized source contributions. Case N grows from 187 cells and
170 loading intervals at 5\% to 663 cells and 704 intervals at 1\%,
reducing the verified error bound from 2.125\% to 0.481\%.
Case P reaches refinement level 11, compared with 7 in S and O and
8 in N. These local extremes distinguish the resolution demands of
separated pulses, narrow localization, and overlapping loads.

\begin{figure}
\centering\includegraphics[width=\linewidth]{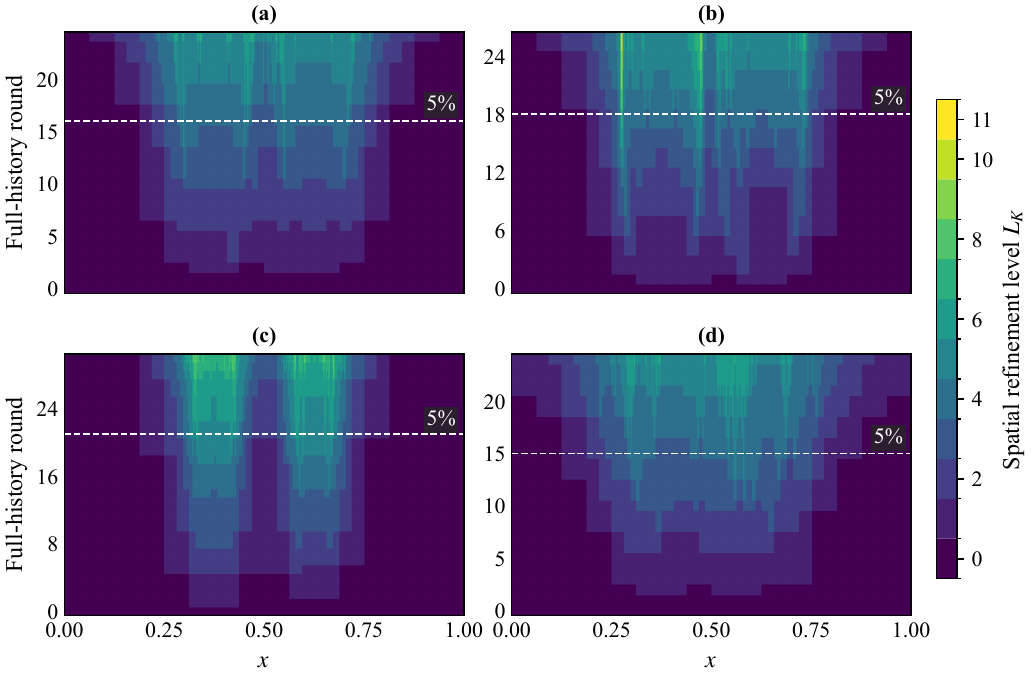}
\caption{Spatial refinement histories: (a) S; (b) P; (c) N; (d) O.
Colors denote $L_K=\log_2[(1/16)/h_K]$ on the saved meshes.
Dashed lines mark 5\% acceptance; the last strips give the 1\% meshes.}
\label{fig:resolution}
\end{figure}

The paired loading orders of Fig.~\ref{fig:history} connect these field
results to the motivation of the method. Their identical final drive and
the uniqueness of the monotone stationary problem give exactly the same
terminal damage. Their cumulative sources have an enclosed
$L^2$ difference in $[0.0108282522,0.0108622403]$. Normalized by the source norm for the original order, the difference is
between 4.9157\% and 4.9317\%. This redistribution is comparable to the
5\% accuracy target despite identical terminal damage.

\subsection{Mechanism controls, accuracy allocation and computational cost}
The controls in Table~\ref{tab:controls} examine the ingredients required
by the complete certificate. On a coarse sequential-loading mesh, omitting
the spatial contribution leaves an absolute expression
$8.2522\times10^{-3}$, while the source error is enclosed near
$2.901\times10^{-2}$. Relative to the enclosed reference norm, these
values correspond to approximately 3.75\% and at least 13.17\%.
The algebraic omission control retains the zero initial state at every
endpoint; the source error is $0.32$, while the incomplete expression is
$5.7781\times10^{-3}$. A two-endpoint uniform calculation with the loading
term omitted gives the same error and an incomplete expression near
$10^{-12}$. The complete bounds cover and reject all these underresolved histories.
\begin{table}
\centering
\caption{Complete and incomplete source bounds in the channel-omission controls.}
\label{tab:controls}
\small
\setlength{\tabcolsep}{3pt}
\begin{tabular}{p{.27\linewidth}lll}
\toprule
Control & Complete & Incomplete & Verified error interval\\
\midrule
Spatial term omitted & $1.9626\times10^{1}$ & $8.2522\times10^{-3}$ & $[0.0290008,0.0290179]$\\
Algebraic term omitted & $5.7101\times10^{0}$ & $5.7781\times10^{-3}$ & $[0.319999,0.320001]$\\
Loading term omitted & $1.2846\times10^{1}$ & $1.0277\times10^{-12}$ & $[0.319999,0.320001]$\\
Two-endpoint pulse check & $1.8725\times10^{2}$ & $1.2574\times10^{-161}$ & $[0.319999,0.320001]$\\
\bottomrule
\end{tabular}
\par\smallskip\begin{minipage}{.97\linewidth}\footnotesize Intervals are rounded outward. The three uniform errors equal $0.32$ analytically. Figure~\ref{fig:blind} shows subsequent refinement of the two-endpoint pulse calculation.\end{minipage}
\end{table}

The pulse control makes the role of complete loading intervals especially
clear. Its initial partition consists of $t=0$ and $t=1$, where the prescribed
drive vanishes, giving a zero initial raw source. The known polynomial
pulse has an interior maximum and exact cumulative source $0.32$.
Full-interval residual and variation bounds give the large initial
certificate in Table~\ref{tab:controls} and select loading refinement.
Figure~\ref{fig:blind} compares this adaptation with the standard start
that already includes the turning point. Both attain the same 0.773\%
certificate with 410 loading intervals and the original 16 spatial cells.
Full-interval information therefore recovers the initially unresolved pulse.

\begin{figure}
\centering\includegraphics[width=.82\linewidth]{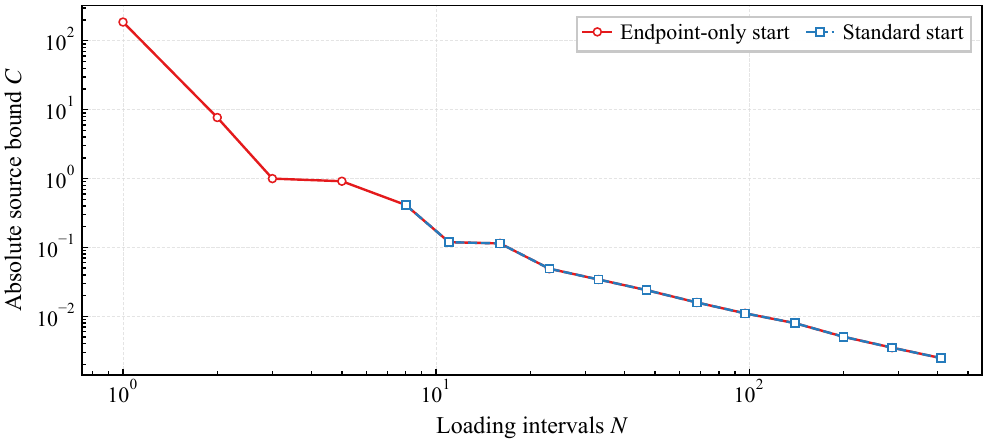}
\caption{Pulse recovery with $M=16$: two-endpoint and breakpoint-resolving
initial partitions. Both histories reach $N=410$ and a 0.773\% relative certificate.}
\label{fig:blind}
\end{figure}

Two controls isolate the history enclosure. Ordinary cubic interpolation
of a nonnegative nodal damage increment attains an interior value
$-0.068276736$ in the reconstruction example. The coefficient-order
constraint in Eq.~\eqref{eq:coeff-limit} ensures a nonnegative reconstructed
increment throughout the cell and supplies the admissibility required by
Theorem~\ref{thm:history}. On the same accepted two-pulse history, the
fixed full-history barrier gives a source bound of $0.0167321$ and the
causal prefix barrier gives $0.0104256$, a reduction of 37.7\%.
On these identical states, the prefix construction avoids enlarging
earlier source contributions with residuals arising later.

In the balanced 1\% runs for S and P, spatial and loading bounds
initially dominate (Figs.~\ref{fig:allocation} and
\ref{fig:partition-history}). Refinement reduces them until the
algebraic bound exceeds its budget fraction, triggering a tenfold
tolerance reduction on the same partitions. Spatial and loading
refinement then resume.

\begin{figure}
\centering\includegraphics[width=\linewidth]{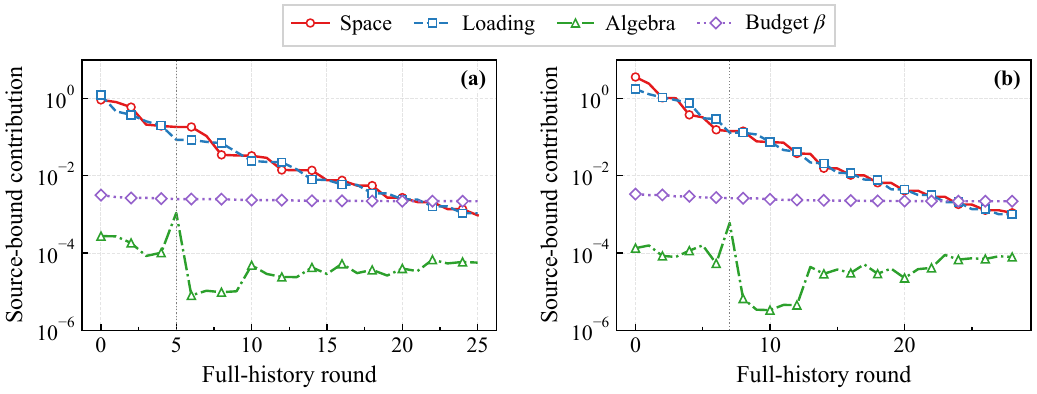}
\caption{Source-bound contributions and admissible budget in the balanced-CG
1\% runs: (a) S; (b) P. Loading denotes $\|B_\tau\|_2+\|B_q\|_2$;
vertical dotted lines mark solver-tolerance reductions.}\label{fig:allocation}
\end{figure}
\begin{figure}
\centering\includegraphics[width=\linewidth]{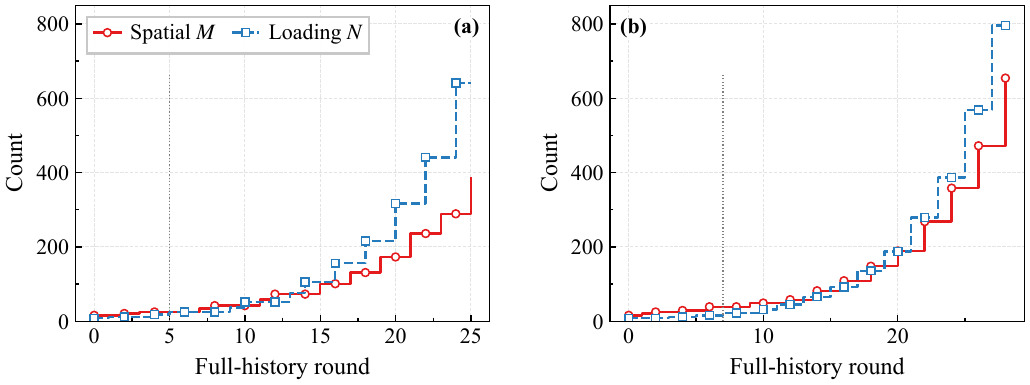}
\caption{Spatial and loading partitions in the balanced-CG 1\% runs:
(a) S; (b) P. Vertical dotted lines mark solver-tolerance reductions,
which retain both partitions.}\label{fig:partition-history}
\end{figure}

Table~\ref{tab:costs} quantifies the resulting allocation at the common
5\% accuracy target. The damage-variation baseline scores cell
$J=[x_j,x_{j+1}]$ by
$\max_n|u_n(x_{j+1})-u_n(x_j)|$, with the same bulk fraction as source
marking. Source marking reduces cumulative matrix--vector node work from
25.71 to 16.48 million for S and from 40.12 to 19.29 million for P,
corresponding to reductions of 35.9\% and 51.9\%, respectively. The final spatial counts
also decrease, from 151 to 131 in S and from 227 to 148 in P. Work sums the participating node counts over all matrix--vector products,
including replays and restricted solves.
\begin{table}
\centering
\caption{Computational cost at the common 5\% target, including all history replays.}
\label{tab:costs}
\small
\begin{tabular}{llrrrr}
\toprule
Case & Marking / solver accuracy & $M/N$ & Work ($10^6$) & Solves & Time (s)\\
\midrule
S & Source / balanced & 131 / 216 & 16.48 & 6426 & 20.66\\
S & Uniform / tight & 128 / 256 & 4.51 & 616 & 10.97\\
S & Damage variation / balanced & 151 / 218 & 25.71 & 5837 & 20.62\\
S & Source / tight & 131 / 156 & 5.81 & 930 & 12.28\\
P & Source / balanced & 148 / 188 & 19.29 & 4469 & 20.83\\
P & Uniform / tight & 256 / 256 & 36.86 & 1680 & 22.26\\
P & Damage variation / balanced & 227 / 194 & 40.12 & 5548 & 29.69\\
P & Source / tight & 148 / 188 & 23.18 & 3255 & 19.66\\
\bottomrule
\end{tabular}
\par\smallskip\begin{minipage}{.97\linewidth}\footnotesize Work sums the participating node counts of all matrix--vector products. Solves counts all restricted linear systems. Times include solution, reconstruction and certification in a shared execution environment. The complete stopping criterion is identical across policies.\end{minipage}
\end{table}

The solver comparison resolves a second allocation effect. For P, balanced
accuracy reduces matrix--vector node work from 23.18 to 19.29 million
relative to tight CG on the same final spatial and loading partitions,
a reduction of 16.8\%. For S, uniform refinement and tight CG require
less cumulative work than the balanced policy. The time column retains
the costs of reconstruction and certification as well as solution, which
explains why the ranking by matrix--vector work can differ from the ranking
by total elapsed time. In the P calculation at 1\%, direct solves give
237.6 seconds including refinement and replay, compared with 620.6 seconds
for balanced CG.

Figure~\ref{fig:coverage} checks all 417 saved histories, comprising
323 primary and 94 additional-case records. For every coarse,
intermediate, and accepted history, the independently enclosed error
upper endpoint lies below the production bound. Coverage therefore
holds throughout all 28 adaptive trajectories. These deterministic,
correlated records complement the 25 fixed-parameter calculations and
the omission controls.

\begin{figure}
\centering\includegraphics[width=0.7\linewidth]{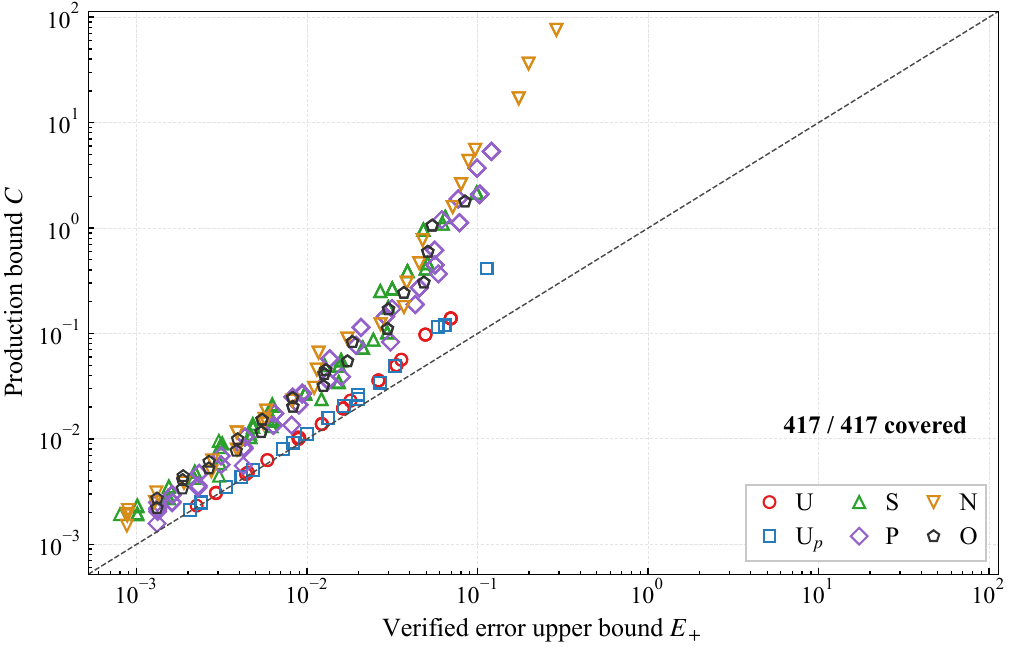}
\caption{Reliability across all 417 saved adaptive histories.
The production bound $C$ covers the independently enclosed error upper endpoint
$E_+$ for every record; the dashed line is $C=E_+$.}\label{fig:coverage}
\end{figure}

Figure~\ref{fig:policy-convergence} compares finite production
certificates against representation size $MN$. Uniform cases retain
their spatial mesh and refine loading; localized cases refine both.
Decreases at unchanged $MN$ identify tighter algebraic solves.
Uniform, Balanced, and Direct denote the three policies specified in
Section~\ref{sec:numerical}, all of which reach the 1\% target.

\begin{figure}[!htbp]
\centering\includegraphics[width=0.9\linewidth]{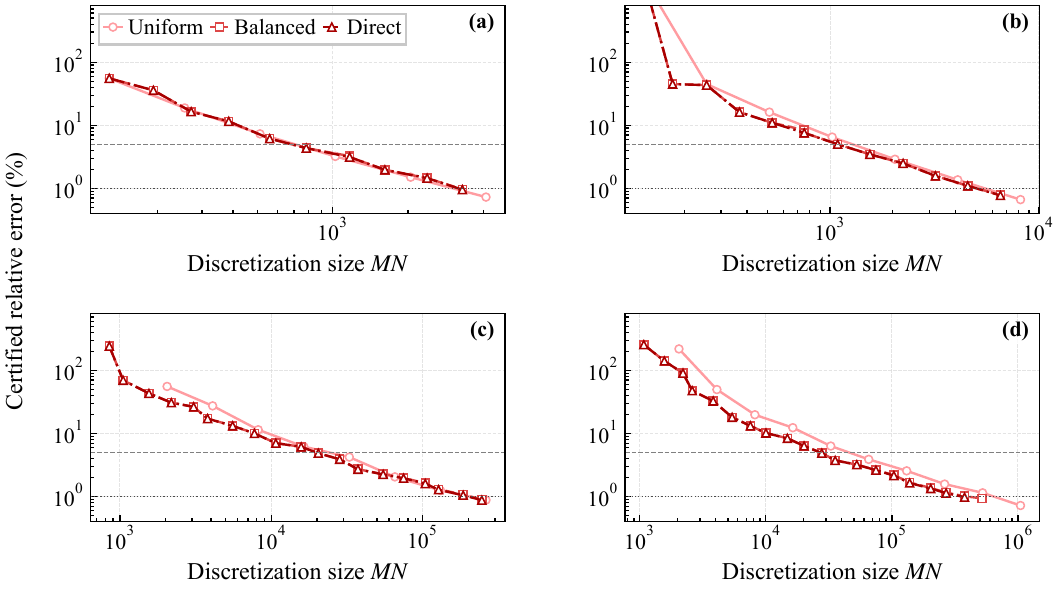}
\caption{Production certificates at the 1\% target: (a) U; (b) U$_p$; (c) S;
(d) P. Uniform, Balanced and Direct identify the three policies described
in the text; horizontal lines mark 1\% and 5\%.}
\label{fig:policy-convergence}
\end{figure}

\label{sec:discussion}

\label{app:results}
Table~\ref{tab:all} gives the accepted resolutions, certificates,
verified errors, and effectivity intervals for all 28 calculations.
Localized reference uncertainty widens the effectivity intervals;
the analytical uniform cases give nearly point-valued intervals.
Here $M$ and $N$ count spatial cells and loading intervals.
\begingroup
\footnotesize
\begin{table}[!htbp]
\centering
\caption{Accepted results for all 28 calculations. Target, Certified, and Verified are percentages; the latter two are upper bounds rounded upward. Effectivity encloses $C/\|Q_{\raw}-Q\|_2$, with endpoints rounded outward.}
\label{tab:all}
\footnotesize
\setlength{\tabcolsep}{3pt}
\begin{tabular}{llrrrrrl}
\toprule
Case & Policy & Target & $M$ & $N$ & Certified & Verified & Effectivity\\
\midrule
S & Source / balanced CG & 1 & 384 & 641 & 0.900 & 0.459 & [1.952, 1.986]\\
S & Source / balanced CG & 5 & 131 & 216 & 4.492 & 1.380 & [3.159, 3.178]\\
S & Source / direct & 1 & 384 & 641 & 0.875 & 0.460 & [1.897, 1.931]\\
S & Source / direct & 5 & 131 & 156 & 4.857 & 2.111 & [2.240, 2.249]\\
S & Uniform / tight CG & 1 & 256 & 1024 & 0.882 & 0.364 & [2.412, 2.465]\\
S & Uniform / tight CG & 5 & 128 & 256 & 4.219 & 1.443 & [2.846, 2.863]\\
P & Source / balanced CG & 1 & 654 & 796 & 0.934 & 0.606 & [1.537, 3.239]\\
P & Source / balanced CG & 5 & 148 & 188 & 4.888 & 1.969 & [2.413, 2.879]\\
P & Source / direct & 1 & 474 & 796 & 0.979 & 0.606 & [1.610, 3.391]\\
P & Source / direct & 5 & 148 & 188 & 4.875 & 1.969 & [2.407, 2.872]\\
P & Uniform / tight CG & 1 & 1024 & 1024 & 0.721 & 0.601 & [1.198, 2.549]\\
P & Uniform / tight CG & 5 & 256 & 256 & 3.883 & 1.938 & [1.965, 2.352]\\
U & Source / balanced CG & 1 & 16 & 206 & 0.968 & 0.912 & [1.061, 1.062]\\
U & Source / balanced CG & 5 & 16 & 49 & 4.346 & 3.811 & [1.134, 1.135]\\
U & Source / direct & 1 & 16 & 206 & 0.954 & 0.913 & [1.044, 1.045]\\
U & Source / direct & 5 & 16 & 49 & 4.346 & 3.811 & [1.134, 1.135]\\
U & Uniform / tight CG & 1 & 16 & 256 & 0.728 & 0.704 & [1.034, 1.035]\\
U & Uniform / tight CG & 5 & 16 & 64 & 3.214 & 2.812 & [1.138, 1.139]\\
U$_p$ & Source / balanced CG & 1 & 16 & 410 & 0.785 & 0.748 & [1.047, 1.048]\\
U$_p$ & Source / balanced CG & 5 & 16 & 97 & 3.462 & 3.135 & [1.101, 1.102]\\
U$_p$ & Source / direct & 1 & 16 & 410 & 0.773 & 0.750 & [1.030, 1.031]\\
U$_p$ & Source / direct & 5 & 16 & 97 & 3.462 & 3.135 & [1.101, 1.102]\\
U$_p$ & Uniform / tight CG & 1 & 16 & 512 & 0.663 & 0.644 & [1.029, 1.030]\\
U$_p$ & Uniform / tight CG & 5 & 16 & 128 & 2.895 & 2.573 & [1.121, 1.122]\\
N & Source / direct & 1 & 663 & 704 & 0.829 & 0.481 & [1.718, 1.751]\\
N & Source / direct & 5 & 187 & 170 & 4.447 & 2.125 & [2.045, 2.055]\\
O & Source / direct & 1 & 387 & 400 & 0.823 & 0.488 & [1.681, 1.682]\\
O & Source / direct & 5 & 121 & 96 & 4.396 & 1.992 & [2.156, 2.157]\\
\bottomrule
\end{tabular}
\par\smallskip
\end{table}

\endgroup

\subsection{Two-dimensional loading-order mechanism}\label{sec:2d-mechanism}

The two-dimensional experiment measures loading-order redistribution of
the cumulative source for histories with common terminal damage and global work.
It uses $\Omega=(0,1)^2$, $a=1$, $b=1/64$, $T=2$, zero initial damage
and homogeneous Neumann conditions. Two nonnegative loading regions are
defined by
\begin{equation}\label{eq:2d-shapes}
 f_A(x,y)=512x^2(1-x)^4y^2(1-y)^2,\qquad
 f_B(x,y)=512x^4(1-x)^2y^2(1-y)^2.
\end{equation}
Their maxima are centered at $(1/3,1/2)$ and $(2/3,1/2)$.
Set $\alpha(t)=\min(t,1)$ and $\beta(t)=\max(t-1,0)$ on $[0,2]$.
The two drives are $f^{AB}=\alpha f_A+\beta f_B$ and
$f^{BA}=\beta f_A+\alpha f_B$. Both drives vary in two spatial coordinates and share the initial state
and final drive.

\begin{proposition}[Common terminal damage and global source]\label{prop:2d-endpoint}
The two continuous obstacle histories just defined have exactly the same
terminal damage and exactly the same global cumulative source.
\end{proposition}
\begin{proof}
\emph{Stationary solutions realize the monotone obstacle history.}
The coercive bilinear form $A_f$ gives a unique stationary solution
$z_f\in H^1$ with $A_f(z_f,\phi)=(f,\phi)$ for every $\phi\in H^1$.
The negative-part and $(z_f-1)_+$ tests give $0\leq z_f\leq1$.
For $g\geq f$, subtraction of the two stationary equations gives
\begin{equation}\label{eq:stationary-order-identity}
 \LL_g(z_g-z_f)=(g-f)(1-z_f)\geq0.
\end{equation}
With $\eta=(z_f-z_g)_+$, the weak equation yields
\[
 0\leq\|\eta\|_0^2
 \leq A_g(z_f-z_g,\eta)
 =-((g-f)(1-z_f),\eta)\leq0.
\]
Thus $z_g\geq z_f$. Inductively, $d_{n-1}^\pi=z_{f(t_{n-1})}$ and
monotonicity make $z_{f(t_n)}$ admissible at step $n$. Stationary
minimality and the obstacle minimization give
\[
 \EE_{t_n}(z_{f(t_n)})
 \leq\EE_{t_n}(d_n^\pi)
 \leq\EE_{t_n}(z_{f(t_n)}).
\]
Strict convexity therefore implies $d_n^\pi=z_{f(t_n)}$. Starting
from $z_{f(0)}=0$ and passing to the continuous limit gives
\[
 d(t)=z_{f(t)},\qquad
 d^{AB}(T)=z_{f_A+f_B}=d^{BA}(T).
\]

\emph{Loading regularity needed to test with $d_t$.}
Let $h=z_{f(t)}-z_{f(s)}$. Subtraction at two times, followed by testing
with $h$, yields
\begin{align}
 \|h\|_0^2
 &\leq A_{f(t)}(h,h)
   =((f(t)-f(s))(1-z_{f(s)}),h)\notag\\
 &\leq |\Omega|^{1/2}\|f(t)-f(s)\|_\infty\|h\|_2
 \leq\frac{|\Omega|^{1/2}}{\sqrt a}
       \|f(t)-f(s)\|_\infty\|h\|_0.\label{eq:stationary-h1-chain}
\end{align}
Consequently $\|h\|_0\leq |\Omega|^{1/2}\|f(t)-f(s)\|_\infty/\sqrt a$,
including $h=0$. The Hilbert-space Lipschitz estimate gives
$d\in W^{1,\infty}(0,T;H^1(\Omega))$, so its quadratic energy obeys
the chain rule.

\emph{Global work and the shared endpoints.}
Testing the weak stationary equation with $d_t\in H^1$ gives, at almost
every time,
\begin{align}
 \int_\Omega f(1-d)d_t\dd x
 &=a\int_\Omega dd_t\dd x+b\int_\Omega\nabla d\cdot\nabla d_t\dd x\notag\\
 &=\frac{\dd}{\dd t}\frac12
       \int_\Omega(a d^2+b|\nabla d|^2)\dd x.\label{eq:global-work-chain}
\end{align}
All products are integrable by the $H^1$ loading bound. Fubini's theorem
and integration in loading now give
\begin{equation}\label{eq:global-endpoint-work}
 \int_\Omega Q[d]\dd x
 =\frac12\left[\int_\Omega(a d^2+b|\nabla d|^2)\dd x\right]_0^T.
\end{equation}
The common initial and terminal fields prove equality of the two global
sources. Reflection gives a second argument for this symmetric example.
The map $(x,y)\mapsto(1-x,y)$ preserves the domain, operator and Neumann
condition, and sends $f^{AB}$ to $f^{BA}$. Uniqueness therefore gives the corresponding reflection of $d$ and,
by its defining integral, of $Q$. The two local source fields are therefore
reflections of one another and may differ at the same spatial point.
\end{proof}

For these stationary monotone histories the global cancellation has a
local counterpart in the distributional sense,
\begin{equation}\label{eq:continuum-history-flux}
 Q[d]=\tfrac12\big[a d^2+b|\nabla d|^2\big]_0^T
       -b\,\nabla\!\cdot\!\int_0^T d_t\nabla d\dd t.
\end{equation}
For $\phi\in C_c^\infty(\Omega)$, test with $\phi d_t\in H^1$ and
integrate in loading to obtain
\begin{align}
 \int_\Omega Q[d]\phi\dd x
 &=\int_0^T\!\int_\Omega\{a dd_t\phi
                +b\nabla d\cdot\nabla(\phi d_t)\}\dd x\dd t\notag\\
 &=\int_0^T\!\int_\Omega(a dd_t+b\nabla d\cdot\nabla d_t)\phi\dd x\dd t
     +b\int_0^T\!\int_\Omega d_t\nabla d\cdot\nabla\phi\dd x\dd t\notag\\
 &=\frac12\int_\Omega[a d^2+b|\nabla d|^2]_0^T\phi\dd x
    +b\int_\Omega\left(\int_0^Td_t\nabla d\dd t\right)\cdot\nabla\phi\dd x.
 \label{eq:local-flux-proof-chain}
\end{align}
Multiplication by the fixed bounded $\phi$ gives the weighted quadratic
chain rule in the last line. The accumulated vector field is integrable
because $d_t,\nabla d\in L^2$ in space and are bounded in loading.
Its last pairing is precisely minus the distributional divergence in
Eq.~\eqref{eq:continuum-history-flux}. The common terminal damage fixes the endpoint term, while the
accumulated gradient-work flux retains the local dependence on order.

Nested uniform grids and loading partitions containing $t=1$ test
spatial and loading refinement separately and jointly. We resolve the
order effect when the enclosed $L^2$ source difference has a positive
lower endpoint, and accept a production field using
\eqref{eq:relative-certificate}.

Figure~\ref{fig:2d-drives} shows the two polynomial loading regions.
Their reflected geometry permits exact terminal and global comparisons,
while retaining spatial variation in both coordinates. The computed terminal
fields in Fig.~\ref{fig:2d-damage} agree to a maximum nodal difference
of $4.106\times10^{-14}$, consistent with the exact continuum equality
established in Proposition~\ref{prop:2d-endpoint}.

\begin{figure}\centering\includegraphics[width=.82\linewidth]{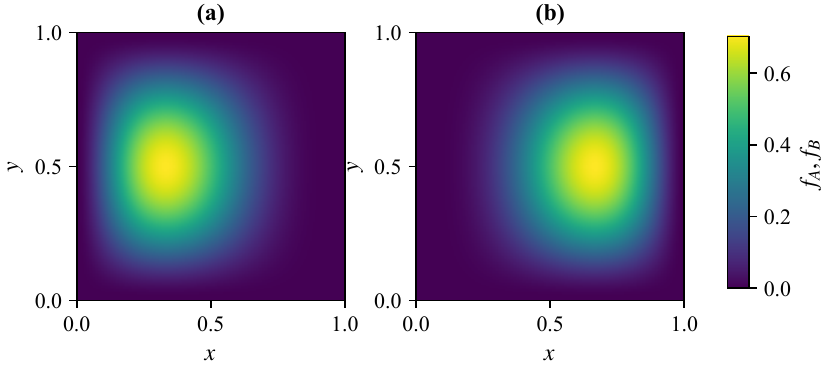}
\caption{Prescribed loading regions on the unit square: (a) $f_A$; (b) $f_B$.
The common scale shows the analytical polynomials in Eq.~\eqref{eq:2d-shapes}.}\label{fig:2d-drives}\end{figure}
\begin{figure}\centering\includegraphics[width=\linewidth]{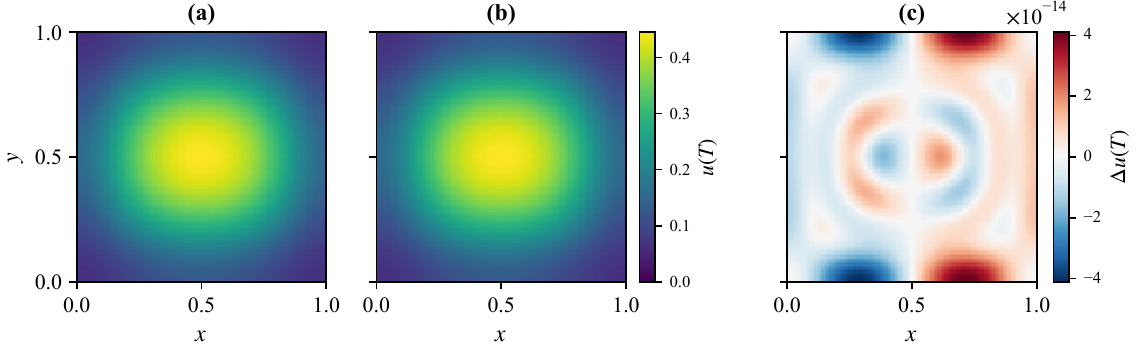}
\caption{Terminal damage at $M=64$, $N=128$: (a) $A\to B$;
(b) $B\to A$; (c) computed difference. Panels (a,b) share a scale;
(c) uses the displayed $10^{-14}$ scale.}
\label{fig:2d-damage}\end{figure}

The local sources remain different (Fig.~\ref{fig:2d-source}).
The independent stationary reference has absolute source radius
$0.0002714729$ and encloses the
continuous source difference in
$[0.0064813,0.0074702]$ in $L^2(\Omega)$.
Relative to the norm of the continuous $A\to B$ source, this gives
$[12.59,14.67]\%$.
The positive lower endpoint resolves the two-dimensional order effect.
The global source curves overlap as predicted by the energy identity
(Fig.~\ref{fig:2d-global}).

\begin{figure}\centering\includegraphics[width=\linewidth]{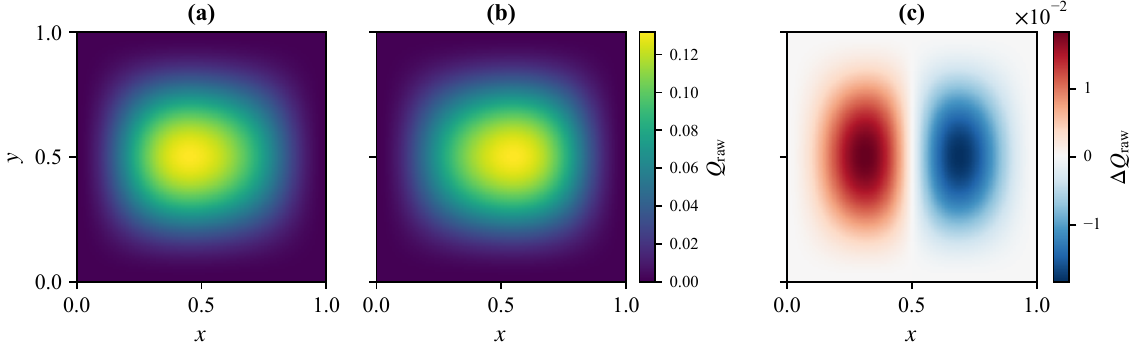}
\caption{Delivered source at $M=64$, $N=128$: (a) $Q_{\raw}^{AB}$;
(b) $Q_{\raw}^{BA}$; (c) their signed difference. Panels (a,b) share
a physical scale; (c) uses a centered scale.}\label{fig:2d-source}\end{figure}
\begin{figure}\centering\includegraphics[width=.7\linewidth]{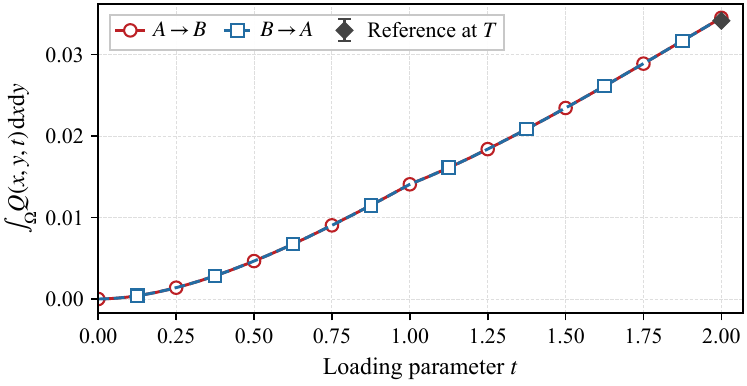}
\caption{Global endpoint accumulations at $M=64$, $N=128$ for both loading orders.
The diamond and error bar enclose the continuous reference integral at $T$.}
\label{fig:2d-global}\end{figure}

Table~\ref{tab:2d} reports all eleven completed production calculations.
The cell count is $M^2$. Times include solution and certification in the
shared execution environment, with reference construction and replay
checks recorded separately.
In the spatial scan, $M=16,32,64,128$ at $N=64$; in the loading scan,
$N=32,64,128,256$ at $M=64$. The paired orders are also solved directly
on three progressively refined grid/partition combinations.
At these three paired resolutions, the relative difference between the
stored sources is approximately 14.183\%, 13.836\% and 13.709\%, respectively.
These numerical-field differences approach the independently enclosed
continuum separation reported above as both spatial and loading resolutions
increase.

Every production error interval is independently enclosed and covered by
its production bound. The $M=64$, $N=128$ calculation passes the 5\%
reference-free test with a relative certificate below
4.076\%, while independent
verification bounds its actual relative error by
1.522\%.
\begin{table}[!htb]
\centering\small
\caption{Controlled two-dimensional calculations: production certificates, independent verification, effectivity and observed computation times.}\label{tab:2d}

\begin{tabular}{llrrrrrr}
\toprule
Order & $M$ & $N$ & $C$ & $C^{\rm rel}$ (\%) & $E_+^{\rm rel}$ (\%) & Effectivity & Time (s)\\
\midrule
AB & 16 & 32 & 0.0247722 & 87.441 & 4.319 & $[11.26,14.80]$ & 1.2\\
AB & 16 & 64 & 0.0232676 & 80.729 & 2.308 & $[19.79,35.36]$ & 1.9\\
AB & 32 & 64 & 0.0069848 & 15.453 & 2.485 & $[5.51,9.45]$ & 25.4\\
AB & 64 & 32 & 0.0038171 & 7.726 & 4.550 & $[1.64,2.14]$ & 56.3\\
AB & 64 & 64 & 0.0025966 & 5.234 & 2.531 & $[2.01,3.42]$ & 102.6\\
AB & 64 & 128 & 0.0020249 & 4.076 & 1.522 & $[2.61,8.04]$ & 225.1\\
AB & 64 & 256 & 0.0017497 & 3.521 & 1.019 & $[3.37,55.17]$ & 476.7\\
AB & 128 & 64 & 0.0014671 & 2.891 & 2.542 & $[1.13,1.92]$ & 484.4\\
BA & 16 & 32 & 0.0247722 & 87.441 & 4.319 & $[11.26,14.80]$ & 1.0\\
BA & 32 & 64 & 0.0069848 & 15.453 & 2.485 & $[5.51,9.45]$ & 19.4\\
BA & 64 & 128 & 0.0020249 & 4.076 & 1.522 & $[2.61,8.04]$ & 202.7\\
\bottomrule
\end{tabular}
\end{table}

Spatial refinement reduces the spatial majorant toward the loading
floor; loading refinement reduces the endpoint contribution toward the
spatial floor (Fig.~\ref{fig:2d-refinement}). Signed spatial and endpoint
errors can cancel, producing nonmonotone actual errors in the
fixed-loading scan.

\begin{figure}[!htb]\centering\includegraphics[width=0.7\linewidth]{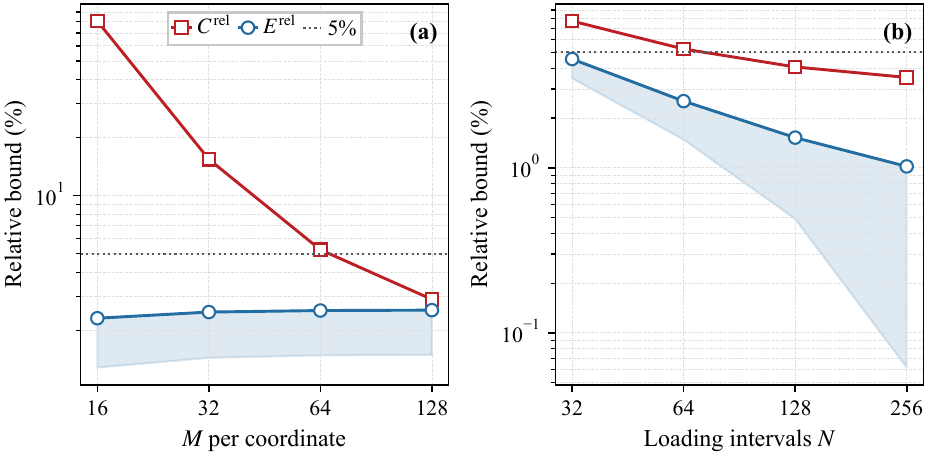}
\caption{Two-dimensional refinement: (a) space at $N=64$; (b) loading at $M=64$.
$C^{\rm rel}$ is the production certificate; the $E^{\rm rel}$ band
encloses the independently verified error.}
\label{fig:2d-refinement}\end{figure}

The reference uses $128\times128$ cells and 128 loading intervals,
with the stationary enclosure and exact polynomial source integration
of Section~\ref{sec:2d-certificate}. Reflection supplies the reverse
order with the same uncertainty. Replays of both accepted $64\times64$
histories reproduce the delivered fields bit for bit and confirm
coefficient order and boundary equalities throughout.
\subsection{Mixed-mode loading history in a notched rock specimen}\label{sec:sen}

\begin{figure}[!htb]\centering
\includegraphics[width=.82\linewidth]{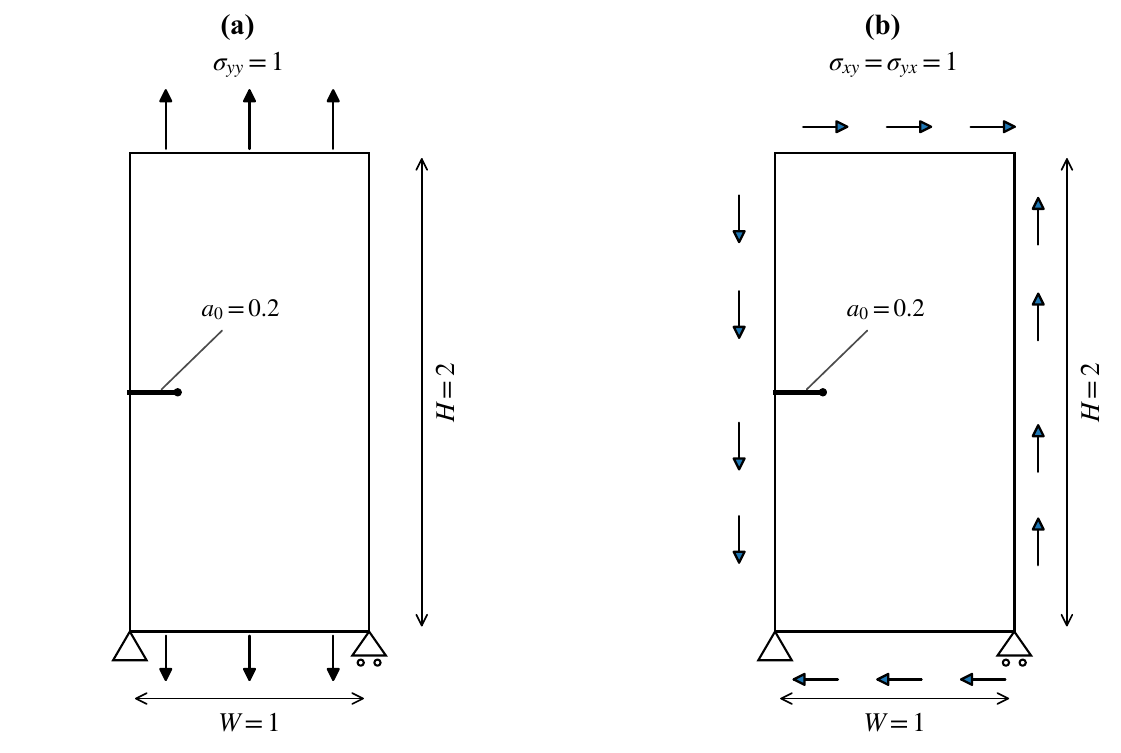}
\caption{SEN elastic preprocessing: (a) unit opening; (b) unit shear.
The slit is traction free, and the indicated constraints remove rigid motions.}
\label{fig:sen-geometry}\end{figure}

The single-edge-notched (SEN) specimen connects source certification to
sequential opening and shear of a rock-motivated geometry. Linear
elasticity generates the prescribed drive in \eqref{eq:energy}.
Reversing the loading order at a common final load tests its effect on
the cumulative source and the cost of certifying that source to 5\%.

The elastic preprocessing domain is a rectangle of width $W=1$ and height
$H=2$, with a traction-free horizontal slit $[0,0.2]\times\{1\}$.
Plane-strain isotropic elasticity uses nondimensional $E=1$ and $\nu=0.2$.
Two unit solutions, $\boldsymbol w_{\mathrm I}$ and
$\boldsymbol w_{\mathrm{II}}$, are computed on an $80\times160$
bilinear finite-element grid. Their exterior tractions are
$\boldsymbol\sigma^{\mathrm I}\boldsymbol n$ and
$\boldsymbol\sigma^{\mathrm{II}}\boldsymbol n$, where
\begin{equation}\label{eq:sen-tractions}
 \boldsymbol\sigma^{\mathrm I}=\begin{pmatrix}0&0\\0&1\end{pmatrix},
 \qquad
 \boldsymbol\sigma^{\mathrm{II}}=\begin{pmatrix}0&1\\1&0\end{pmatrix}.
\end{equation}
The shear tractions act on all four exterior edges and balance both force
and moment. At $(0,0)$ both displacement components are fixed, and at
$(1,0)$ the vertical component is fixed; these three constraints remove
rigid motions. Duplicated slit-face nodes share a single tip node.
Figure~\ref{fig:sen-geometry} specifies the geometry and loading. The
damage problem retains the complete structured rectangle with homogeneous
Neumann conditions; the notch enters through the prescribed drive.

Let $G_{ij}=\tfrac12\boldsymbol\varepsilon(\boldsymbol w_i):
\mathbb C:\boldsymbol\varepsilon(\boldsymbol w_j)$, with
$i,j\in\{\mathrm I,\mathrm{II}\}$. A nonnegative Gaussian recovery
of width $0.1W$, followed by bilinear interpolation on a fixed
$16\times32$ background grid, defines the stored energy Gram field
$\overline{\boldsymbol G}$. Its diagonal and determinant bounds are
checked outward at every stored vertex. Bilinear interpolation preserves
positive definiteness by convexity, so the prescribed density
\begin{equation}\label{eq:sen-drive}
 f(x,y,t)=\gamma\bigl[
 \alpha(t)^2\overline G_{11}(x,y)
 +2\alpha(t)\beta(t)\overline G_{12}(x,y)
 +\beta(t)^2\overline G_{22}(x,y)\bigr]
\end{equation}
is nonnegative. The signed interaction term is retained
(Fig.~\ref{fig:sen-driving}). The recovered Gram field remains fixed during all damage refinements.
A separate elasticity refinement from $80\times160$ to $160\times320$
changes it by approximately 0.312\% in relative Frobenius norm.
This measures mechanical-input sensitivity; the source certificate
concerns the fixed stored field.

\begin{figure}[!htb]\centering
\includegraphics[width=.95\linewidth]{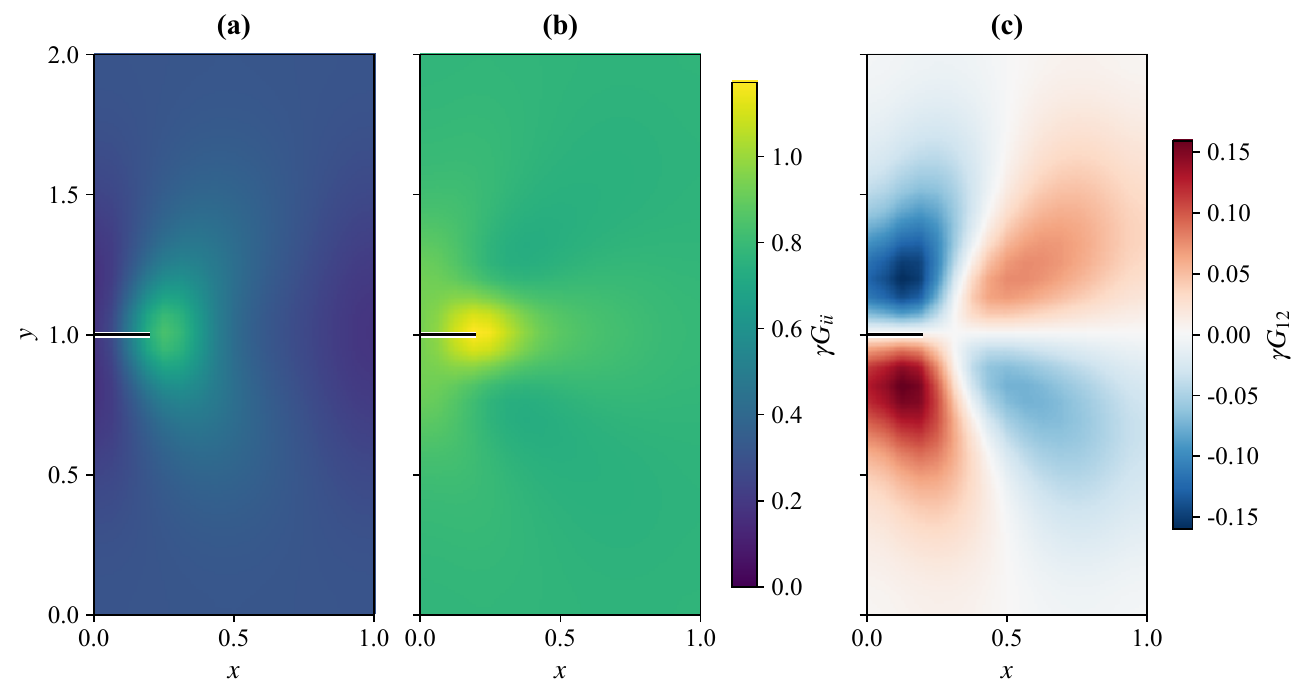}
\caption{Stored energy fields: (a) opening, $\gamma\overline G_{11}$;
(b) shear, $\gamma\overline G_{22}$; (c) signed interaction,
$\gamma\overline G_{12}$. Panels (a,b) share a scale.}
\label{fig:sen-driving}\end{figure}

The damage parameters are $a=1$, $b=2^{-8}$ and $d(0)=0$, giving the
interaction length $\sqrt{b/a}=W/16$. A declared scan considers
$b\in\{1/256,1/64\}$ and terminal peak drives $0.5,1,2,4$ for both
orders. The geometric interaction length is fixed first; the peak drive
closest to a maximum damage of $0.6$ is then selected. This gives a peak
drive of $2$ and $\gamma=0.6395008660428658$. Selection follows these predefined mechanical criteria, independently
of the measured source difference, and all sixteen scan records are retained. With $S(s)=s^3(10-15s+6s^2)$ on $[0,1]$, extended by
zero and one outside this interval, the two paths on $t\in[0,1]$ are
\begin{equation}\label{eq:sen-paths}
 (\alpha_A,\beta_A)=(S(2t),S(2t-1)),\qquad
 (\alpha_B,\beta_B)=(S(2t-1),S(2t)).
\end{equation}
Path A applies opening before shear, and Path B reverses that order
(Fig.~\ref{fig:sen-loading}). Both paths end exactly at $(1,1)$ and
therefore share the final elastic load and prescribed drive. The signed
interaction field $\overline G_{12}$ allows locally nonmonotone driving
despite increasing amplitudes. Both paths are therefore computed through their full obstacle histories.

\begin{figure}[!htb]\centering
\includegraphics[width=.88\linewidth]{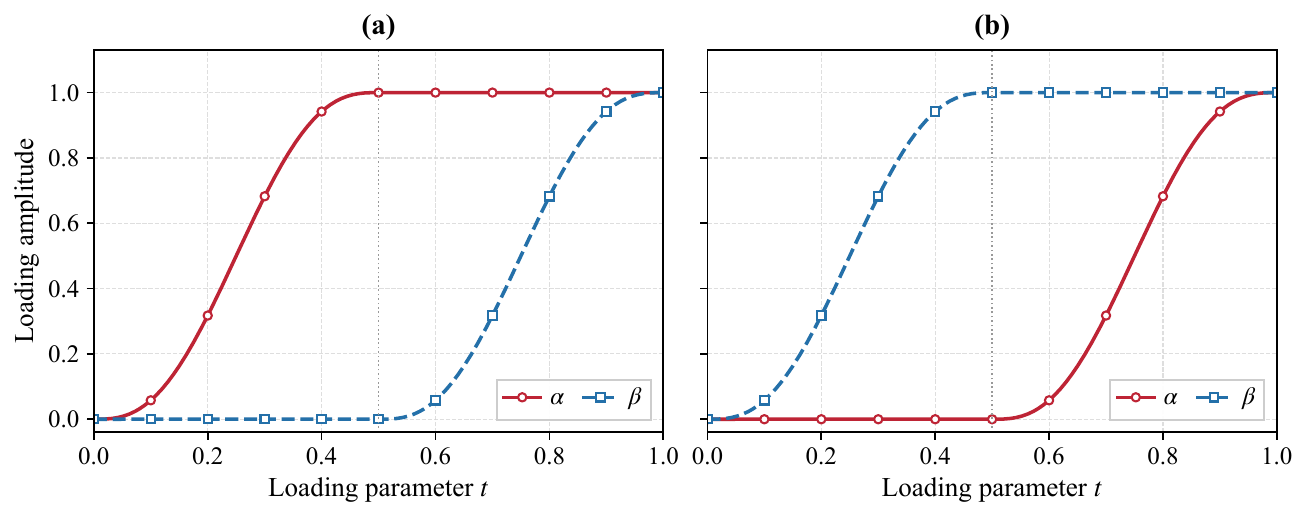}
\caption{Loading amplitudes: (a) opening then shear (Path A);
(b) shear then opening (Path B). Both terminal amplitude pairs equal $(1,1)$.}
\label{fig:sen-loading}\end{figure}

The certification begins with the constant causal provider. Let
$\rho_{j,K,n}$ bound residual channel $j$ on a complete cell and loading
interval, and let $\underline a_n>0$ be a certified lower bound for
$a+f$ throughout $\Omega\times I_n$. A spatially constant radius is
\begin{equation}\label{eq:sen-constant}
 e_{j,n}^{\mathrm C}
 =\max_{m\le n}\frac{\max_K\rho_{j,K,m}}{\underline a_m}.
\end{equation}
It is nonnegative and nondecreasing, and
$\mathcal L_t e_{j,n}^{\mathrm C}\ge\rho_{j,K,n}$ on the interval.
The tensor reconstruction retains every positive coefficient increment.
Cells are frozen exactly when all supporting coefficients are unchanged,
preserving multiplier complementarity.

An auxiliary elliptic response uses residual localization to construct
a spatially varying uncertainty radius. For each channel, set $\rho_{j,K}=\max_n\rho_{j,K,n}$ and
construct a continuous bilinear $g_j$ by assigning the maximum adjacent
$\rho_{j,K}$ to each vertex. Then $g_j\ge\rho_{j,K}$ throughout each
cell. The exact auxiliary response satisfies
\begin{equation}\label{eq:sen-auxiliary}
 (-b\Delta+a)w_j=g_j,\qquad \partial_n w_j=0.
\end{equation}
A numerical solve supplies a $C^2$ Neumann spline $z_j$.
The residual correction \eqref{snum:aux-correction} encloses the exact
supersolution $w_j$ in Bernstein arithmetic.

For Path B, a two-stage provider additionally uses the positive drive
after the first mode is fully applied. Bernstein bounds of
Eq.~\eqref{eq:sen-drive} give $a+f\ge a_*$ throughout the second
stage, with the certified value $a_*\ge1.65819725$.
Let $E_{0,j}$ be the first-stage constant prefix radius at $t=1/2$.
An auxiliary obstacle solve provides a Neumann spline candidate $z_j$.
Applying the stage correction \eqref{snum:stage-correction} with
all extrema and divisions enclosed outward gives
\begin{equation}\label{eq:sen-stage-correction}
 R_j=g_j+b\Delta z_j-a_*z_j,\quad
 \delta_j=\max\left\{
 \frac{\|(R_j)_+\|_\infty}{a_*},\,
 E_{0,j}-\inf_\Omega z_j,\,0\right\},\quad
 v_j=z_j+\delta_j.
\end{equation}
The correction gives $v_j\ge E_{0,j}\ge0$ and
$(-b\Delta+a_*)v_j=g_j-R_j+a_*\delta_j\ge g_j$.
Since $a+f\ge a_*$, it follows that
$\mathcal L_t v_j\ge g_j\ge\rho_j$. Joining the first-stage
prefix radii to this fixed second-stage supersolution preserves the
causal inheritance condition. The independently valid constant cap
$\max\{E_{0,j},\max_K\rho_{j,K}/a_*\}$ can also be used; the pointwise minimum of these two valid error enclosures gives a
sharper history bound for source transfer. Source transfer uses $E_{0,j}$ in the first stage and the sharpened
enclosure in the second, on the same computed history.

Every loading partition contains $t=1/2$. The quadratic products of
the quintic clocks are enclosed over complete intervals in degree-ten
Bernstein form. With affine loading interpolation $\mathcal I_n$ and
the spline history $p$, the exact residual identity is
\begin{equation}\label{eq:sen-loading-residual}
 f(1-p)-\mathcal I_n[f(1-p)]
 =(f-\mathcal I_n f)(1-p)
   +\theta(1-\theta)\Delta f_n\Delta p_n.
\end{equation}
Tensor-product Bernstein bounds enclose both terms throughout the
quintic loading interval. For the bilinear damage history interpolated affinely
on one loading interval, set $\Delta u=u_1-u_0$ and
$m_k=\int_0^1\theta^k f(\theta)\dd\theta$. Direct integration gives
the exact defect of the endpoint contribution,
\begin{equation}\label{eq:sen-endpoint}
 f_1(1-u_1)\Delta u-Q[u]\big|_{I_n}
 =\Delta u\bigl[(f_1-m_0)(1-u_0)-(f_1-m_1)\Delta u\bigr].
\end{equation}
Moments, products, and $L^2$ integration over physical area $WH=2$
are evaluated outward. Acceptance requires $q_->C$ and
$C/(q_--C)\leq0.05$.

Table~\ref{tab:sen-certification} and Fig.~\ref{fig:sen-certification}
show the progression from the original constant barrier to accepted
computations. At $32\times64$ cells and $N=64$, the two constant-barrier
certificates are 23.196\% and 27.115\%, with decimal bounds rounded
upward. Using the auxiliary provider on the same saved histories reduces
them to 8.640\% and 13.490\%. The numerical damage and $Q_{\raw}$
are identical in each such pair; this reduction measures sharper
certification. Subsequent joint refinement to $64\times128$ and
$128\times256$ cells, with $N=128$ and $256$, reduces the final
certificates to 3.256\% and 4.951\%. Both paths satisfy the 5\% test. The table reports these certificates
alongside independently enclosed actual-error intervals.

\begin{table}[!htbp]
\centering\small
\setlength{\tabcolsep}{5pt}
\caption{SEN certification for opening then shear (Path A) and the reverse order (Path B). C and A denote constant and auxiliary providers. The first two rows of each path use identical states and $Q_{\raw}$. Bounds are rounded outward.}
\label{tab:sen-certification}
\begin{tabular}{llrlrrr}
\toprule
Path & Grid & $N$ & Barrier & $C^{\rm rel}$ (\%) & $E^{\rm rel}$ (\%) & CPU (s) \\
\midrule
A & $32\times64$ & 64 & C & \textbf{23.196} & [2.024, 4.831] & 17.4 \\
A & $32\times64$ & 64 & A & 8.640 & [2.024, 4.831] & 19.8 \\
A & $64\times128$ & 128 & A & 4.956 & [0.826, 3.656] & 166.6 \\
A & $128\times256$ & 256 & A & \textbf{3.256} & [0.291, 3.144] & 873.3 \\
\midrule
B & $32\times64$ & 64 & C & \textbf{27.115} & [1.120, 7.734] & 21.1 \\
B & $32\times64$ & 64 & A & 13.490 & [1.120, 7.734] & 9.8 \\
B & $64\times128$ & 128 & A & 6.482 & [0.000, 5.782] & 177.9 \\
B & $128\times256$ & 256 & A & \textbf{4.951} & [0.000, 4.830] & 1281.2 \\
\bottomrule
\end{tabular}
\par\smallskip
\end{table}

\begin{figure}[!htb]\centering
\includegraphics[width=\linewidth]{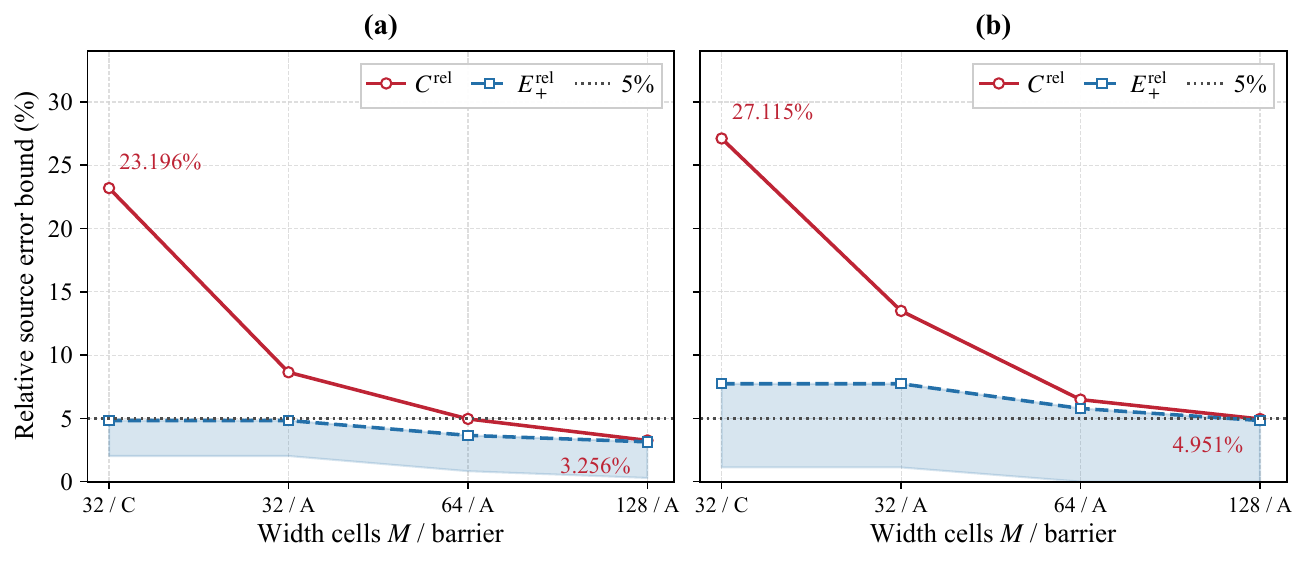}
\caption{SEN certification: (a) Path A; (b) Path B. C and A denote
constant and auxiliary barriers; $N=2M$. The first two points use the
same history. Shading retains reference uncertainty, and the dotted line
marks 5\%.}
\label{fig:sen-certification}\end{figure}

The spatial and loading studies retain the same frozen mechanical input.
The relative difference between the two computed terminal damage fields
is approximately 0.0736\%, 0.0716\% and 0.0710\% at the three
successive paired resolutions. The corresponding delivered-source
differences are 2.881\%, 2.693\% and 2.635\%, normalized by Path A.
Figures~\ref{fig:sen-damage} and~\ref{fig:sen-source} display the
accepted fields and their differences. The small terminal-damage
difference therefore coexists with a substantially larger cumulative-source
difference. At the accepted resolution the global delivered sources are
approximately $0.281169$ and $0.282028$, a difference of
$8.594\times10^{-4}$. The reference enclosure below tests separation of the continuous sources.

\begin{figure}[!htb]\centering
\includegraphics[width=.95\linewidth]{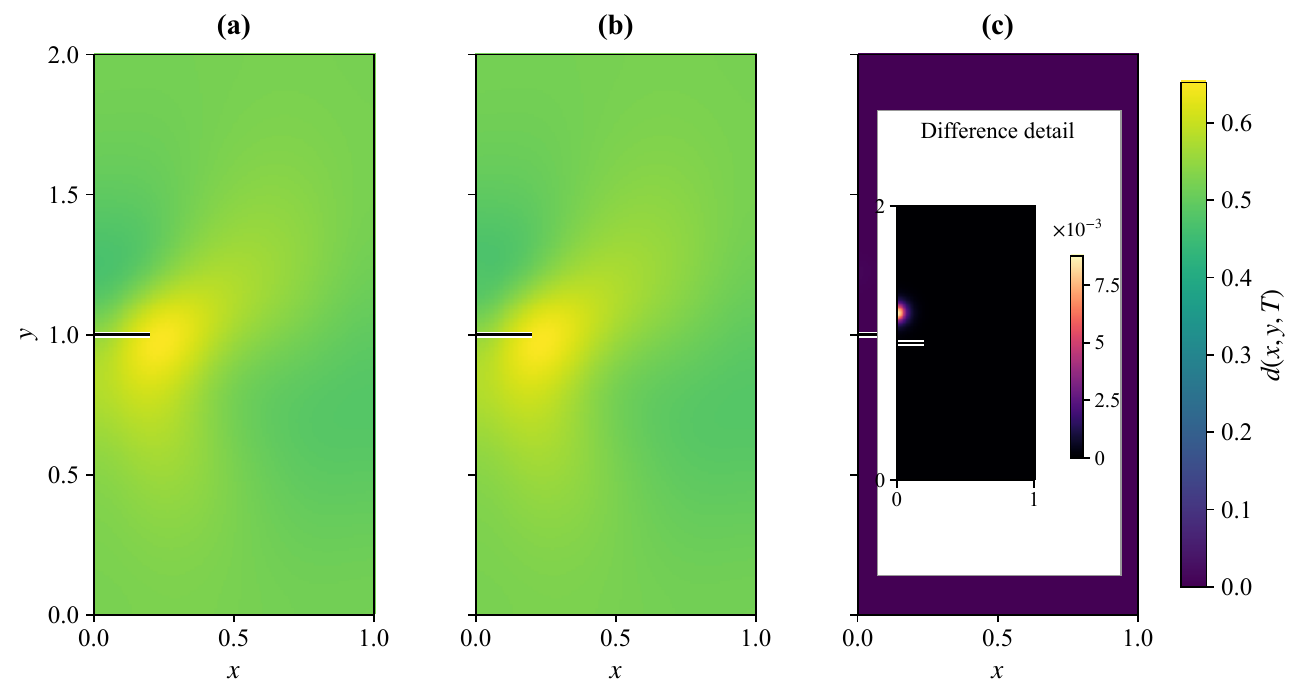}
\caption{SEN terminal damage: (a) Path A; (b) Path B; (c) absolute
difference. Main panels share a scale; the inset resolves the small
difference on its indicated scale.}
\label{fig:sen-damage}\end{figure}
\begin{figure}[!htb]\centering
\includegraphics[width=.95\linewidth]{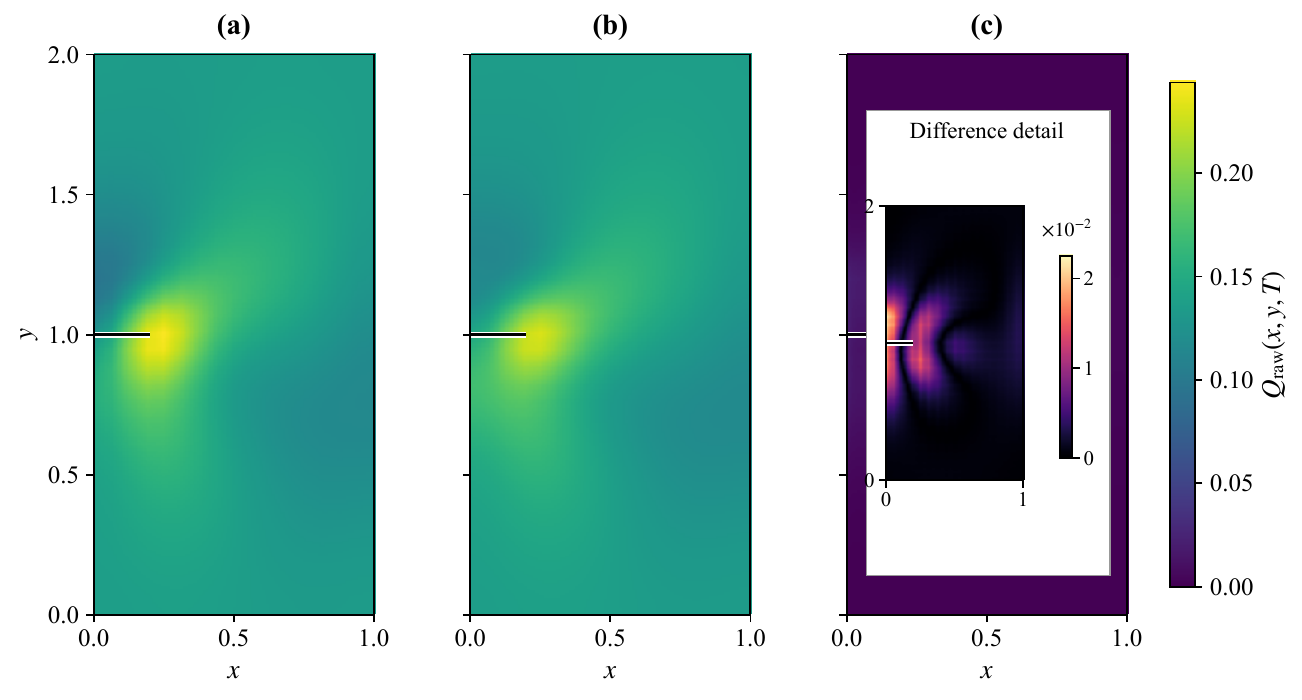}
\caption{SEN delivered source: (a) Path A; (b) Path B; (c) absolute
difference. Main panels share a scale; the inset uses its indicated
difference scale.}
\label{fig:sen-source}\end{figure}

Each reference uses separate state histories on $128\times256$ cells
and 512 loading intervals, direct sparse obstacle solves, and continuous
polynomial source integration. Production and reference share the
reconstruction and interval-arithmetic routines; each bound is computed
from its own residuals. The accepted relative-error intervals are
$[0.291,3.144]\%$ and $[0,4.830]\%$, both covered by the production
bounds. All 22 retained certificate records from nine distinct production
discretizations satisfy coverage, including alternative enclosures and
reconstruction controls on shared histories. Tolerance and contact-guard
outcomes, including rejected runs, remain in the data. CPU times include
the base solve and any retained re-enclosure cost.

Adding the two reference radii to the exact polynomial comparison of
their source centers yields
\begin{equation}\label{eq:sen-source-separation}
 \frac{\|Q[d_A]-Q[d_B]\|_{L^2(\Omega)}}
      {\|Q[d_A]\|_{L^2(\Omega)}}
 \in[0.7234,8.3276]\%.
\end{equation}
The local difference remains positive after reference uncertainty is
included. The signed global difference lies in $[-0.013036,0.012968]$,
which contains zero. Figure~\ref{fig:sen-uncertainty} shows the accepted
pointwise bounds.

\begin{figure}[!htb]\centering
\includegraphics[width=.70\linewidth]{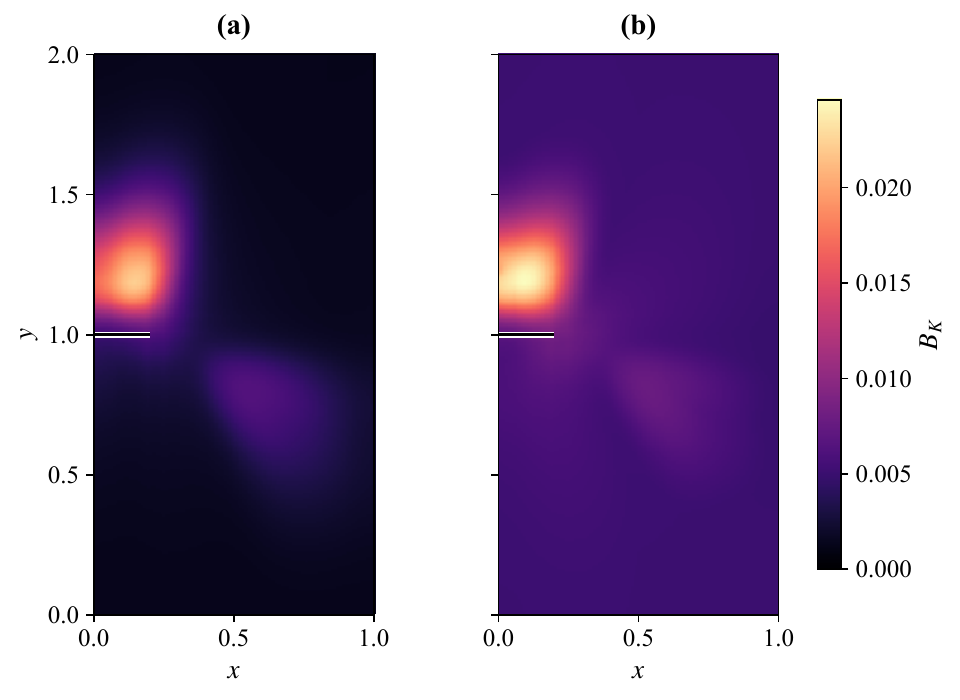}
\caption{Complete cellwise source-uncertainty bounds for the accepted
SEN calculations: (a) Path A; (b) Path B, on a common scale.}
\label{fig:sen-uncertainty}\end{figure}

Damage concentrates near the recovered notch-tip field, with diffuse
damage from the positive background drive. Nearly coincident terminal
fields coexist with distinct cumulative sources, and the auxiliary
providers bring both outputs within the prescribed 5\% accuracy target
after refinement.
\FloatBarrier

\section{Conclusions}\label{snum:conclusions}
We have established a causal a posteriori bound for the cumulative
energy field of scalar irreversible gradient damage. Nondecreasing
supersolutions enclose uncertainty inherited through the obstacle, and
a primitive identity transfers this history bound to the output.
The endpoint defect and outward arithmetic complete the certificate for
the delivered field. Admissible splines, Bernstein bounds, and constant
or elliptic supersolutions give computable realizations in one dimension
and on two-dimensional rectangles. The analysis also identifies
sufficient conditions for certificate decay.

The numerical evidence addresses both reliability and usefulness. All
417 saved adaptive histories and 25 fixed-parameter calculations have
independently enclosed errors below their production bounds. The narrow-localization
case decreases from approximately 158.78\% relative error to a verified
upper error of 2.125\% at the 5\% certified target. Source-oriented
marking saves 35.9\% and 51.9\% of accumulated matrix--vector work in
the two comparisons with damage-variation marking. The two-dimensional
mechanism separates local-history effects from identical terminal and
global quantities, while the mechanical benchmark demonstrates the
constant-to-auxiliary-to-refined certification sequence.

Further developments concern certified driving-density errors for
displacement--damage coupling, admissible reconstructions and pointwise
supersolutions on general meshes, and efficient history reuse. Establishing
convergence of the adaptive scores would complement the delivered-output
guarantee proved here.

\section*{Acknowledgment}
This work was supported by the National Key R\&D Program of China
(No. 2024YFB2606101), the National Natural Science Foundation of China
(No. 52478418), and the Sichuan Science and Technology Program
(No. 2025ZNSFSC1294).
\section*{Declarations}
Data and the implementation are available on request. The authors declare
no conflict of interest.
\bibliographystyle{ws-m3as}
\bibliography{references}
\end{document}